\documentclass[12pt,onecolumn]{IEEEtran}%
\usepackage{amsfonts}
\usepackage{amssymb}
\usepackage{amsmath}
\usepackage{amscd}
\usepackage{graphicx}
\usepackage{setspace}
\usepackage{multirow}%

\usepackage{color}

\usepackage{soul}

\providecommand{\U}[1]{\protect\rule{.1in}{.1in}}
\newtheorem{thm}{Theorem}[section]

\newtheorem{prop}[thm]{Proposition}

\newtheorem{rem}{Remark}[section]
\begin{document}

\title{Relay Channel with Orthogonal Components and Structured Interference Known at
the Source}
\author{\normalsize{Ka\u{g}an~Bakano\u{g}lu$\dagger$ \hbox{ } Elza~Erkip$\dagger$ \hbox{ } Osvaldo~Simeone$\ddagger$ \hbox{ } Sholomo Shamai (Shitz)$\diamond$}
\thanks{$\dagger$ Dept. of Electrical and Computer Engineering, Polytechnic Institute of New York University, Brooklyn, NY 11201, USA (email:%
kbakan01@students.poly.edu, elza@poly.edu).}
\thanks{$\ddagger$ Center for Wireless Communications and Signal Processing Research, New Jersey Institute of Technology, Newark, New Jersey%
07102-1982, USA (email: osvaldo.simeone@njit.edu).}
\thanks{$\diamond$Dept. of Electrical Engineering, Technion Institute of Technology, Technion City,
Haifa 32000, Israel (email: sshlomo@ee.technion.ac.il).}}%
\date{}
\maketitle
\doublespacing
\vspace{-0.7cm}
\begin{abstract}
A relay channel with orthogonal components that is affected by an interference signal that is
non-causally available only at the source is studied. The interference signal has
structure in that it is produced by another transmitter communicating with its
own destination. Moreover, the interferer is not willing to adjust its
communication strategy to minimize the interference. Knowledge of the
interferer's signal may be acquired by the source, for instance, by exploiting
HARQ retransmissions on the interferer's link. The source can then utilize the
relay not only for communicating its own message, but also for cooperative
interference mitigation at the destination by informing the relay about the
interference signal. Proposed transmission strategies are based on partial decode-and-forward (PDF) relaying and leverage the
interference structure. Achievable schemes are derived for discrete memoryless
models, Gaussian and Ricean fading channels. Furthermore, optimal strategies are
identified in some special cases. Finally, numerical results bring insight
into the advantages of utilizing the interference structure at the source, relay or destination.

\end{abstract}

\pagenumbering{arabic}
\vspace{-0.9cm}
\section{Introduction}

Interference provides a major impairment for many current and envisioned
wireless systems. Techniques that are able to mitigate interference are thus
expected to be of increasing importance in the design of wireless networks.
Two critical features of interfering signals can be leveraged to make the task
of interference management more effective. The first is that interference is
\textit{structured}, as it typically arises from the transmissions of other wireless users. The second is that\textit{ information about the interference} can be
obtained by wireless nodes in the vicinity of the interferer in a number of
relevant scenarios. As an example, assume that the interferer employs
retransmissions (HARQ) on its link. A\ node in the vicinity may be able to
decode a prior retransmission and use this information in order to facilitate
interference mitigation. Another scenario where interference information is
conventionally assumed is cognitive radio.

In this paper, we investigate interference mitigation strategies for a
cooperative communication scenario in which a source communicates via an out of band relay
to a destination in the presence of an external interferer. The interferer is
not willing, or not allowed, to change its transmission strategy to reduce
interference on the destination. The source is able to obtain information
about the interferer signal prior to transmission in the current block. We are interested in studying
effective ways to use such interference information at the source, in particular, the ones that leverage the structure of the interference.

The source can exploit the interference structure in a number of ways. For
instance, the structure of the interference signal potentially allows the source
to reduce the amount of spectral resources necessary for communicating
interference information to the relay. A second way to take advantage of the
interference structure is for the source, possibly with the help of the relay,
to help reception of the interfering signal at the destination so that the destination can decode and remove the interference. In
this work, we will explore these possibilities and assess the advantages of
strategies that exploit the interferer's structure with respect to the
techniques studied in \cite{Zaidi_new} that assume an unstructured interferer.
\vspace{-0.5cm}
\subsection{Related Work}

A simple model for the interference signal assumes that it is
unstructured and, in particular, that it consists of an independent
identically distributed (i.i.d.) sequence. This model is accurate, for
instance, if the interference is the sum of the contributions of many
interferers, all of comparable powers. In information-theoretic terms, an
i.i.d. interference can be modelled as the ``state'' of a channel. The capacity
of a state-dependent memoryless channel, where the state sequence (i.e., the
interference) is available non-causally at the transmitter, is established by
Gel'fand Pinsker in \cite{Gelfand} (see also \cite{heegard}). Costa
\cite{Costa} applied Gel'fand and Pinsker's (GP) result to the Additive White
Gaussian Noise (AWGN) model with additive Gaussian state, giving rise to the so called Dirty Paper Coding
(DPC) technique. DPC achieves the state-free capacity even though the state is
not known at the receiver. It was shown in \cite{cohen}, \cite{zamir}, that
this principle continues to hold even if the state is not Gaussian. However when there is no channel state information at the transmitter (CSIT), DPC can no longer achieve state-free capacity for AWGN with additive Gaussian state. This aspect for various assumptions on the channel gains was captured and studied in \cite{Zhang}-\cite{Khina}.

Extensions to the multiuser case were performed by Gel'fand and Pinsker in
\cite{gelfand2} and by Kim \textit{et al.} in \cite{9}\cite{thesis}. In
particular, in \cite{9}\cite{thesis} it is proved that for MACs multi-user versions of GP and DPC,
referred to as multi-user GP (MU-GP) and multi user DPC (MU-DPC) respectively, achieve optimal performance. In
\cite{1}, Somekh-Baruch \textit{et al.} considered a memoryless two-user MAC,
with the state available only to one of the encoders. The capacity region is shown to be obtained
by generalized GP (GGP) and generalized DPC (GDPC). The scenario
studied in this paper, but with an i.i.d. state is investigated in \cite{Zaidi_new}, \cite{Zaidi_EU} for a Discrete Memoryless (DM) and Gaussian relay channels with an in-band relay and \cite{Zaidi_orthogonal} \cite{Kagan_10} for a DM and Gaussian relay channel with an out-of-band relay where lower and upper bounds on the capacity are
derived.


With a single dominating interferer, interference structure can be utilized. This was recognized in \cite{maric}, where a scenario in which a
transmitter-receiver pair communicates in the presence of a single interferer
is studied. It is shown therein that using GP coding, and hence treating the
interference as if it were unstructured, it is generally suboptimal and
\textit{interference forwarding} with joint decoding at the destination can be
beneficial \cite{2}. This aspect is further studied in
\cite{osvaldo} for a MAC with structured interference available at one
encoder, in \cite{Kagan_10} for a Gaussian relay channel with an out-of-band relay and in \cite{liu} for a cognitive Z-interference channel, where extensions of the techniques proposed in \cite{maric} are investigated.
\vspace{-0.5cm}
\subsection{Contributions and Organization}

In this paper, we study interference mitigation techniques for the
relay channel with orthogonal components~\cite{ElGamal} and with an external interferer whose signal is non-causally
available only at the source. The relay channel with orthogonal components model is chosen due to its ability to model half-duplex communications and availability of capacity achieving strategies~\cite{ElGamal}.  We propose several techniques for discrete
memoryless, AWGN and Ricean fading channels that leverage interference structure to different
degrees. We also establish optimality of specific transmission strategies for
several special cases. Finally, numerical results bring insight into the
advantages of interference mitigation techniques that exploit the interference structure.

\vspace{-0.3cm}
\section{System Model}

\label{system_model} The scenario under study consists of a relay channel with
an orthogonal source-to-relay link in the presence of an interferer. In this
model, the source sends two different signals, one to the relay and one to the
destination with the help of the relay in orthogonal channels. The interference signal is available
non-causally to the source as depicted in Fig.
\ref{More_simplified_relay_channel_w_interference_informed_source}. We first
consider a Discrete Memoryless Channel (DMC)\ version of the channel, which is
described by the conditional probability mass functions (pmfs) $P_{Y_{D}%
|X_{SD}X_{R}X_{I}}$ and $P_{Y_{R}|X_{SR}X_{R}}$ where $Y_{D}\in
\mathcal{Y}_{D}$, $Y_{R}\in\mathcal{Y}_{R}$, $(X_{SD},X_{SR})\in
\mathcal{X}_{SD}\times\mathcal{X}_{SR}$, $X_{R}\in\mathcal{X}_{R}$ and
$X_{I}\in\mathcal{X}_{I}$ are the destination ($D$) output, the relay ($R$)
output, the source ($S$) input, the relay ($R$) input and the interference
($I$) signal, respectively. The pmf $P_{Y_{D}|X_{SD}X_{R}X_{I}}$ describes
the stochastic relation between the signals transmitted by the source towards
the destination ($X_{SD}$), by the relay ($X_{R}$), and by the interferer
($X_{I}$) and the signal received at the destination ($Y_{D}$). Similarly, the
pmf $P_{Y_{R}|X_{SR}X_{R}}$ represents the relationship between the signals
transmitted by the source towards the relay ($X_{SR}$) and by the relay
($X_{R})$ and the signal received at the relay ($Y_{R}$).

The source wishes to transmit a message $W$ to the destination with the help
of the relay in $n$ channel uses. The message $W$ is uniformly distributed
over the set $\mathcal{W}=\{1,\ldots,2^{nR}\}$, where $R$ is the rate in
bits/channel use. The interferer employs a fixed (and given) codebook that is
not subject to design. In particular, the codebook of the interferer is
assumed to be chosen by the interfering terminal independently to communicate
with some other destination which is not modeled explicitly. The message
$W_{I}$ of the interferer is assumed to be uniformly distributed over the set
$\mathcal{W_{I}}=\{1,\cdots,2^{nR_{I}}\}$, where $R_{I}$ is the interferer's
rate in bits/channel use. We assume that the interferer's codebook is
generated according to a pmf $P_{X_{I}}$. The generated codebook of the
interferer is known to all nodes. Furthermore, the interferer's message
$W_{I}$ is known to the source. In the sequel we use the standard definitions of achievable rates and probability of error \cite{Cover}.


We also consider the AWGN scenario shown in Fig.
\ref{More_simplified_relay_channel_w_interference_informed_source_awgn}. For
this model, the input and output relations at time instant $i$ are given as
\vspace{-0.4cm}
\begin{align}
Y_{R,i}   =h_{SR,i}X_{SR,i}+Z_{R,i} \label{gaussian_model} \hbox{ } \hbox{ }\hbox{ and } \hbox{ } \hbox{ } Y_{D,i}  =h_{SD,i}X_{SD,i}+h_{RD,i}X_{R,i}+h_{I,i}X_{I,i}+Z_{D,i}%
\end{align}
\vspace{-1.3cm}\newline
where the noises $Z_{D,i}$ and $Z_{R,i}$ are independent zero mean complex
Gaussian random variables with unit variance, and $h_{SR,i}$, $h_{SD,i}$,
$h_{RD,i}$ and $h_{I,i}$ are the complex valued channel gains accounting for
propagation from the source to the relay ($h_{SR,i}),$ from the source to the
destination ($h_{SD,i}$), from the relay to the destination ($h_{RD,i}$), and
from the interferer to the destination ($h_{I,i}$), respectively.

The codewords of the source $X_{SR}^{n}$ and $X_{SD}^{n}$ are
subject to a total energy constraint $nP_{S}$ and the codewords of the relay
$X_{R}^{n}$ is subject to power constraint $nP_{R}$ as
 \vspace{-0.4cm}
\begin{subequations}
\label{power_constraints}%
\begin{align}
\frac{1}{n}\sum_{i=1}^{n}\left(\mathbb{E}\left[  |X_{SR,i}|^{2}\right]
+\mathbb{E}\left[  |X_{SD,i}|^{2}\right]\right)   &  \leq P_{S} \hbox{ } \hbox{ } \hbox{ and } \hbox{ } \hbox{ } \frac{1}{n}\sum_{i=1}^{n}\mathbb{E}\left[  |X_{R,i}|^{2}\right]    \leq
P_{R}.\label{relay_power_constraint}%
\end{align}  \vspace{-1.2cm}\newline
We assume that the interferer codebook is generated i.i.d. with complex Gaussian distribution
with zero mean and power $P_{I}$. We use the notation $\mathcal{C}%
(x)=\log_{2}(1+x)$.

For the AWGN model (\ref{gaussian_model}), we study the following two
scenarios: (\textit{i}) \emph{No fading}: All channel gains remain constant
over the entire coding block and are perfectly known to all
nodes; (\textit{ii}) \emph{Ergodic fading}: All channel gains change in an
ergodic fashion. The instantaneous values of channel gains are not known to
the transmitters but are available at the receivers. Specifically, $h_{SR}$ is
known at the relay and $h_{SD}$, $h_{RD}$, $h_{I}$ are known at the
destination. Channel statistics instead are known at all nodes. In particular,
we assume that channel gains $h_{SR}$, $h_{SD}$, $h_{RD}$ and $h_{I}$ are
independent Ricean distributed with parameters $K_{SR}$, $K_{SD}$, $K_{RD}$,
$K_{I},$ respectively.

 \vspace{-0.2cm}
\section{Achievable Rates for DM and AWGN Channels\label{ach_rates}}
While reference \cite{Zaidi_new} focuses on achievable rates for the case
where the interference signal $X_{I}^{n}$ is i.i.d., here we concentrate on
techniques that exploit the interference structure, as modeled in the
previous section. The advantages of leveraging interference structure will be
discussed in Sec. \ref{num_res} via numerical results through comparison with
the techniques proposed in \cite{Zaidi_new} (which will be also recalled below
for completeness).

The proposed techniques are based on the following considerations. In
\cite{ElGamal}, El Gamal and Zaidi prove the optimality of partial
decode-and-forward (PDF) for the relay channel with orthogonal components in
Fig. \ref{More_simplified_relay_channel_w_interference_informed_source}
without interference. Motivated by this, we assume that the relaying strategy
for the source message is based on PDF. Specifically, the source message $W$
is split into two independent messages, $W=(W^{\prime},W^{\prime\prime})$,
where $W^{\prime}$ is sent through the relay and $W^{\prime\prime}$ is sent
directly to the destination. The messages $W^{\prime}$ and $W^{\prime\prime}$
are uniformly distributed over the set $\mathcal{W^{\prime}}=\{1,\cdots
,2^{nR^{\prime}}\}$ and $\mathcal{W^{\prime\prime}}=\{1,\cdots,2^{nR^{\prime
\prime}}\}$, respectively, and the total rate is $R=R^{\prime}+R^{\prime
\prime}$.

Interference mitigation is utilized either by the source only or by both the
source and the relay in a cooperative fashion. In order to perform cooperative
interference mitigation, the source needs to share the interference
information with the relay. The structure of the interference plays an
important role for the two phases of informing the relay of the interference
and of interference mitigation towards the destination. We categorize the
possible strategies in both phases as follows:
\end{subequations}
\begin{itemize}
\item \emph{Communication of interference to the relay}: When the source
chooses to inform the relay about the interfering signal, it has two options:

\begin{enumerate}
\item \emph{Digital interference sharing}: The structure of the interference
is exploited as follows. The source encodes the interference index $W_{I}$
into a codebook (not necessarily the same as the interferer's codebook) and
sends it to the relay through the orthogonal source-relay ($S-R$) channel. The
relay then decodes the interference index $W_{I}$.

\item \emph{Compressed interference sharing}: The structure of the
interference is not used and the interference is treated as an i.i.d.
sequence. Specifically, the source simply quantizes the interference sequence
$X_{I}^{n}$ and forwards the compressed description to the relay through the
orthogonal source-relay channel. The relay hence recovers the interference
sequence with some quantization distortion.
\end{enumerate}

\item \emph{Interference mitigation at the destination}: There are several
interference mitigation scenarios applicable to our model depending on the
availability of interference information at the relay. We mainly concentrate
on two approaches:

\begin{enumerate}
\item \emph{Structured approach}: The structure of the interference is
exploited at the destination to decode and remove the interference signal.
Decoding can be facilitated by having the source and/or the relay forward
information about the interference to the destination. When the source does
not inform the relay about the interfering signal, interference forwarding is
performed by the source only. Otherwise, interference forwarding is done
jointly by the source and the relay. In the AWGN\ channel, interference
forwarding is performed by the source and/or relay by transmitting signals
that are coherent with the interferer's signal, so that the correlation
between transmitted signal and interference is positive;

\item \emph{Unstructured approach}: The structure of the interference is
ignored at the destination and the interference is treated as an i.i.d. state.
Interference precoding via GP, MU-GP or GGP for the DMC\ model, and DPC,
MU-DPC and GDPC for the AWGN model, are utilized by the source only or by the
source and the relay jointly depending on the availability of interference
information at the relay. This class of techniques was extensively explored in
\cite{Zaidi_new} and will be considered here only in combination with the
digital approach mentioned above (not applicable in the unstructured model of
\cite{Zaidi_new}), and for reference.
\end{enumerate}
\end{itemize}

Below, we list proposed achievable schemes based on the above categorization.

We only consider the scenario where the source and the relay cooperate for both source signal and interference mitigation. Strategies for which the source uses the relay only for signal forwarding, but not for interference management, are the special cases of the schemes below.

\subsubsection{Scheme (D,U) (Digital interference sharing, Unstructured
approach)}

In this scheme, the source sends the interference digitally to the relay, so
that the relay is fully informed about the interference sequence. In addition,
the source also forwards part of the source message to the relay according to
PDF. Then, the source and the relay follow the unstructured approach by
jointly employing multi-user GP (MU-GP) \cite{thesis} to forward the source message.

\begin{prop}
\label{prop(2,D,U,DM)} For Scheme (D,U), the following rate is achievable for
the DM model: \vspace{-0.4cm}
\begin{equation}
R_{(D,U)}=\max_{{}}\min\left\{
\begin{array}
[c]{ll}%
I(U_{S};Y_{D}|U_{R})-I(U_{S};X_{I}|U_{R})+(I(X_{SR};Y_{R}|X_{R})-R_{I})^{+}, &
\\
I(U_{S}U_{R};Y_{D})-I(U_{S}U_{R};X_{I}) &
\end{array}
\right.  \label{(2,D,U,DM)}%
\end{equation} \vspace{-1cm}\newline
where the maximization is taken over the input pmfs $P_{U_{S}U_{R}X_{R}X_{SR}X_{SD}%
|X_{I}}$ of the form $P_{X_{SR}|X_{R}X_{I}}$\newline $P_{U_{S}U_{R}X_{R}%
X_{SD}|X_{I}},$ where $U_{S},U_{R}$ are finite-alphabet auxiliary random
variables.
\end{prop}

\textit{Sketch of the proof}: The message $W$ is split into two messages
$W^{\prime}$ and $W^{\prime\prime}$. The source conveys the message
$W^{\prime}$ to the relay together with interference index $W_{I}$ which leads
to the constraint $R^{\prime}\leq I(X_{SR};Y_{R}|X_{R})-R_{I}$. Since both the
source and the relay have the interference knowledge, they are able to
implement MU-GP \cite{thesis} to send $W^{\prime}$ and $W^{\prime\prime}$ to
the destination. Note that unlike \cite{thesis}, here the two encoders (source
and relay) have the common message $W^{\prime},$ so that the channel from the
source and the relay to the destination is equivalent to the state
(interference) dependent MAC with common message and informed encoders. An
achievable rate region can be derived by following similar steps in
\cite{thesis}\cite{1}, obtaining
 \vspace{-0.5cm}\begin{subequations}
\label{rate_MU-GP_MAC}%
\begin{align}
R^{\prime\prime}\leq I(U_{S};Y_{D}|U_{R})-I(U_{S};X_{I}|U_{R})\\
R^{\prime}+R^{\prime\prime}\leq I(U_{S}U_{R};Y_{D})-I(U_{S}U_{R};X_{I}) &
\end{align}
 \vspace{-1.5cm} \newline for some distribution $P_{U_{S} U_{R}X_{R}X_{SD}|X_{I}}$. Incorporating
(\ref{rate_MU-GP_MAC}) with the constraint on $R^{\prime}$ gives us (\ref{(2,D,U,DM)}).\hfill$\Box$
\end{subequations}
\begin{prop}
\label{prop(2,U,D,GC)} For Scheme (D,U), the following rate is achievable for
the AWGN model (\ref{gaussian_model}):
\vspace{-0.4cm} \begin{align}
R_{(D,U)}=\displaystyle\max_{\substack{\rho_{W^{\prime}},\rho_{W^{\prime
\prime}},\gamma:\\|\rho_{W^{\prime}}|,|\rho_{W^{\prime\prime}}|,\gamma
\in\lbrack0,1]}} &  \min\left\{
\begin{array}
[c]{ll}%
\mathcal{C}\left(  P_{W^{\prime\prime}}\right)  +(\mathcal{C}\left(
|h_{SR}|^{2}(1-\gamma)P_{S}\right)  -R_{I})^{+}, & \\
\mathcal{C}\left(  P_{W^{\prime\prime}}+P_{W^{\prime}}\right)   &
\end{array}
\right.  \label{(2,D,U,GC)}\\
&  \hbox{ subject to }|\rho_{W^{\prime}}|^{2}+|\rho_{W^{\prime\prime}}%
|^{2}\leq1\nonumber
\end{align}
\vspace{-1.2cm} \newline
where $P_{W^{\prime}}=(|h_{RD}|\sqrt{P_{R}}%
+|h_{SD}||\rho_{W^{\prime}}|\sqrt{\gamma P_{S}})^{2}$ and $P_{W^{\prime
\prime}}=|h_{SD}|^{2}|\rho_{W^{\prime\prime}}|^{2}\gamma P_{S}$.
\end{prop}

\textit{Sketch of the proof}: The result is obtained from (\ref{(2,D,U,DM)}),
where all inputs are chosen according to Gaussian distribution. Specifically,
$X_{SD}$ is allocated power $\gamma P_{S}$, $0\leq\gamma
\leq1$, and the remaining power $(1-\gamma)P_{S}$ is allocated to $X_{SR}$. We set $X_{SD}=\rho_{W^{\prime}}\sqrt{\gamma P_S} U_{W^{\prime }}+\rho_{W^{\prime \prime}}\sqrt{\gamma P_S} U_{W^{\prime \prime}}$ and $X_{R}=\sqrt{P_R} U_{W^{\prime}}$ where $U_{W^{\prime}}$ and $U_{W^{\prime \prime}}$ are independent, zero mean, unit variance, complex Gaussian random variables and carrying the messages $W^{\prime}$ and $W^{\prime \prime}$, respectively. Furthermore, $U_{W^{\prime}}$ and $U_{W^{\prime \prime}}$ are independent of $X_I$. The source conveys
$W^{\prime}$ to the relay at rate $R^{\prime}\leq(\mathcal{C}\left(
|h_{SR}|^{2}(1-\gamma)P_{S}\right)  -R_{I})^{+}$ and the interference at rate $R_I$. MU-DPC is used by the source and the relay for
transmission to the destination, where the precoding is done via $U_{S}$
and $U_{R}$ in (\ref{(2,D,U,DM)}) which are chosen to be linear combinations of
($X_{SD}$,$X_{I}$) and ($X_{R}$,$X_{I}$) as $U_{S}=X_{SD}+\alpha_{S}X_{I}$ and
$U_{R}=X_{R}+\alpha_{R}X_{I}$ with inflation factors $\alpha_{S}$ and
$\alpha_{R}$ and $(X_{SD}$,$X_{R})$ jointly complex Gaussian and independent of
$X_{I}$. When the inflation factors are optimized the
effect of the interference is completely eliminated at the destination similar
to \cite{thesis}, leading to (\ref{(2,D,U,GC)}). We refer the readers to \cite{Costa} and \cite{thesis} for details on DPC and MU-DPC. \hfill$\Box$

\begin{rem}
It is shown that the interference-free capacity region can be achieved by
MU-DPC in \cite{9} for Gaussian relay channel when the interference is
non-causally available at both the source and the relay. Apart from the fact
that we consider a relay channel with orthogonal components, the main
difference with \cite{9} is that the relay does not know the interference a
priori but is informed about the interference through the orthogonal
source-relay link. Note that the structure of the interference is essential in
Proposition \ref{prop(2,U,D,GC)} in conveying the interference signal to the
relay. However, this structure is not used in interference mitigation at the destination.
\end{rem}

\begin{rem}
Once can also consider a scheme (D,S) in which the interference is digitally transmitted to the relay and the structured approach for decoding at the destination is
used. Scheme (D,S) may lead to performance improvements over Scheme (D,U) for a DMC. However, for AWGN channels, Scheme (D,S) is inferior to Scheme
(D,U), since Scheme (D,U) is able to completely remove the effect
of interference at the destination via MU-DPC. We will observe in Sec.
\ref{ergodic fading} and Sec. \ref{Sec_num_ef} that, however, for fading channels MU-DPC typically fails to eliminate the effect of the interference at the destination completely and Scheme (D,S) may outperform Scheme (D,U).
\end{rem}

\subsubsection{Scheme (C,U) (Compressed interference sharing, Unstructured
approach)}

With this scheme, studied in \cite{Zaidi_new} and \cite{Zaidi_orthogonal} for the general relay channel and relay channel with orthogonal components
respectively, the source sends the compressed interference signal and the part of the message to the relay and the unstructured
approach is utilized for decoding at the destination. Achievable rate for Scheme (C,U) for our DM model can be obtained from [1, Corollary 1]. It can be extended to
Gaussian case by using an approach similar to [1, Theorem 6] and taking the
complex channel gains into account. The achievable rate for (C,U) for the AWGN model (\ref{gaussian_model}) can be written as
\vspace{-0.5cm}
\begin{align} \label{(C,U,GC)}
R_{(C,U)}=&\displaystyle\max_{\substack{r_{q},\rho_{W^{\prime}},\rho
_{W^{\prime\prime}},\rho_{W_{I}},\gamma:\\|\rho_{W^{\prime}}|,|\rho_{W^{\prime\prime}}|,|\rho_{W_{I}}|,\gamma\in
\lbrack0,1]  }}   \min\left\{(\mathcal{C}\left(
|h_{SR}|^{2}(1-\gamma)P_{S}\right)  -r_{q})^{+}, \mathcal{C}\left(P_{W^{\prime}}\right)\right\} + \mathcal{C}\left(P_{W^{\prime \prime}}\right)\\
&\hbox{ subject to }   0\leq r_{q}\leq\mathcal{C}\left(  |h_{SR}|^{2}%
(1-\gamma)P_{S}\right) \hbox{ and }  |\rho_{W^{\prime}}|^{2}+|\rho_{W^{\prime\prime}}|^{2}+|\rho_{W_{I}}%
|^{2}\leq1,\nonumber
\end{align}
\vspace{-1.2cm} \newline
where $P_{W^{\prime}}=(|h_{RD}|\sqrt{P_{R}}%
+|h_{SD}||\rho_{W^{\prime}}|\sqrt{\gamma P_{S}})^{2}/ \left(1+\xi^2 D+P_{W^{\prime\prime}}\right)$, $P_{W^{\prime\prime}}=|h_{SD}|^{2}|\rho_{W^{\prime\prime}}|^{2} \gamma P_{S}$, $D=P_I2^{-r_q}$ and $\xi=|h_{I}|-|h_{SD}||\rho_{W_I}|\sqrt{\gamma P_S /P_I}$.

\begin{rem}\label{rem_CU_special_case}
When $r_q=0$ in (\ref{(C,U,GC)}), (C,U) boils down to the special case in which the relay is utilized only for source message cooperation and the source mitigates the interference by itself.
\end{rem}
\subsubsection{Scheme (C,S) (Compressed interference sharing, Structured
approach)}
We propose two schemes in the class (C,S). For both schemes, the source informs the relay using compressed interference information, and the structured approach is used to mitigate interference at the destination. The schemes differ in the way the compressed interference information is used at the source, relay and destination nodes.
In the first scheme, referred to as (C,S,1), the compressed interference information is used only to improve the reception of the interference signal at the destination by forwarding an ``analog'' version of the compressed interference. In the second scheme, referred to as (C,S,2), the compressed interference information is re-encoded by source and relay and decoded at the destination in a similar way as for standard compress-and-forward protocols for the relay channel (See, e.g., \cite{ElGamal2}).

\begin{prop}
\label{prop(2,C,S,DM)1} For Scheme (C,S,1), the following rate is achievable for
the DM model:
\vspace{-0.3cm}
\begin{equation}
R_{(C,S,1)}=\max\min\left\{
\begin{array}
[c]{l}%
I(V;Y_{D}|UX_{I})+(I(X_{SR};Y_{R}|X_{R})-I(X_{I};\hat{X}_{I}))^{+},\\
(I(VX_{I};Y_{D}|U)-R_I)^+ + (I(X_{SR};Y_{R}|X_{R})-I(X_{I};\hat{X}_{I}))^{+},\\
I(VU;Y_{D}|X_{I}),\\
(I(VUX_{I};Y_{D})-R_{I})^{+}%
\end{array}
\right.  ,\label{(2,C,S,DM)_1}%
\end{equation} \vspace{-0.8cm}\newline
where the maximum is over all input pmfs $P_{UV\hat{X}_{I}X_{R} X_{SR}%
X_{SD}|X_{I}}$ of the form $P_{\hat{X}_{I}|X_{I}}P_{X_{SR}|X_{R}}P_{U}P_{X_{R}%
|U\hat{X}_{I}}$\newline $P_{V|UX_{I}}P_{X_{SD}|V\hat{X}_{I}}$.
\end{prop}

\textit{Sketch of the proof}: The source quantizes the interference signal
$X_{I}^{n}$ into a reconstruction sequence $\hat{X}_{I}^{n}$ at rate
$I(X_{I};\hat{X}_{I})$ using some test channel $P_{\hat{X}_{I}|X_{I}}$ and
sends the index of the quantized interference and $W^{\prime}$ to the
relay. The relay recovers $\hat{X}_{I}^{n}$ and $W^{\prime}$ successfully if
$R^{\prime}+I(X_{I};\hat{X}_{I})\leq I(X_{SR};Y_{R}|X_{R})$. As a result of the source-to-relay communication, the channel to the destination can be seen as a MAC with common messages in which the source and the relay have the message sets $(W^{\prime},W^{\prime\prime},W_I)$ and $(W^{\prime})$, respectively. The source and relay can thus employ a code in which the source codeword $V^n$ depends on messages $(W^{\prime},W^{\prime\prime},W_I)$ and the relay codeword $U^n$ depends on message $W^{\prime}$. The reason for using auxiliary codebooks instead of the actual transmitted signals $X_{SD}^n$ and $X_R^n$ is because unlike the corresponding conventional model, here the source and the relay also share the compressed interference information $\hat{X}_I^n$. In the scheme (C,S,1) at hand, this information is forwarded in an ``analog'' fashion to the receiver. This is accomplished by mapping the codewords $V^n$ and $U^n$, obtained as discussed above, and the compressed state information $\hat{X}_I^n$, into the transmitted signals $X_{SD}^n$ and $X_R^n$, respectively. Following the results for MAC with common messages \cite{6} \cite{wolf} \cite{deniz}, an achievable rate region is obtained as
\vspace{-0.5cm}
\begin{subequations}
\label{CS_MAC}%
\begin{align}
R^{\prime\prime} &  \leq I(V;Y_{D}|UX_{I})\\
R^{\prime\prime}+R_{I} &  \leq I(VX_{I};Y_{D}|U) \\
R^{\prime\prime}+R^{\prime} &  \leq I(VU;Y_{D}|X_{I}) \\
R^{\prime\prime}+R^{\prime}+R_{I} &  \leq I(VUX_{I};Y_{D}),
\end{align}
\vspace{-1.3cm} \newline
for some input pmf $P_{UV X_{R}X_{SD}|X_{I}\hat{X}_{I}}= P_{U} P_{X_{R}|U\hat{X}_{I} }P_{V|UX_{I}}P_{X_{SD}|V\hat{X}_{I}}$. Incorporating the constraint on $R^{\prime}$ above into
(\ref{CS_MAC}) gives us (\ref{(2,C,S,DM)_1}). \hfill$\Box$
\end{subequations}
\begin{prop}
For Scheme (C,S,1), the following rate is achievable for the AWGN model
(\ref{gaussian_model}):
\begin{align}
R_{(C,S,1)}=\displaystyle\max_{\substack{r_{q},\rho_{W^{\prime}},\rho
_{W^{\prime\prime}},\\\rho_{W_{I}},\bar{\rho}_{W^{\prime}},\bar{\rho}_{\hat
{W}_{I}},\gamma:\\|\rho_{W^{\prime}}|,|\rho_{W^{\prime\prime}}|,|\rho_{W_{I}%
}|,\\|\bar{\rho}_{W^{\prime}}||\bar{\rho}_{W_{I}}|,\gamma\in
\lbrack0,1]}} &  \min\left\{
\begin{array}
[c]{ll}%
\mathcal{C}\left(  P_{W^{\prime\prime}}\right)  +(\mathcal{C}\left(
|h_{SR}|^{2}(1-\gamma)P_{S}\right)  -r_{q})^{+}, & \\
(\mathcal{C}\left(P_{W^{\prime\prime}}+P_{W_{I}%
} \right)-R_I)^+ +(\mathcal{C}\left(
|h_{SR}|^{2}(1-\gamma)P_{S}\right)  -r_{q})^{+}, & \\
\mathcal{C}\left(  P_{W^{\prime\prime}}+P_{W^{\prime}}\right)  , & \\
(\mathcal{C}\left(  P_{W^{\prime\prime}}+P_{W^{\prime}}+P_{W_{I}%
}\right)  -R_{I})^{+} &
\end{array}
\right.   \nonumber\\
& \hbox{ subject to }   0\leq r_{q}\leq\mathcal{C}\left(  |h_{SR}|^{2}%
(1-\gamma)P_{S}\right),  \label{(2,C,S,GC)} \\
&  |\rho_{W^{\prime}}|^{2}+|\rho_{W^{\prime\prime}}|^{2}+|\rho_{W_{I}}%
|^{2}\leq1 \hbox{ and  }  |\bar{\rho}_{W^{\prime}}|^{2}+|\bar{\rho}_{W_{I}}|^{2}\leq1 \nonumber
\end{align}
\vspace{-1.2cm}\newline
where $P_{W^{\prime}}=(|h_{RD}||\bar{\rho}_{W^{\prime}}|\sqrt{P_{R}}%
+|h_{SD}||\rho_{W^{\prime}}|\sqrt{\gamma P_{S}})^{2}/N_{eq}$, $P_{W^{\prime
\prime}}=|h_{SD}|^{2}|\rho_{W^{\prime\prime}}|^{2}\gamma P_{S}/N_{eq},P_{W_{I}}=(|h_{SD}||\rho_{W_{I}}|\sqrt{\gamma P_{S}}+|h_{RD}||\bar{\rho }_{W_I}|\sqrt{P_{R}(1-2^{-r_q})}+|h_{I}|\sqrt{P_{I}})^{2}/N_{eq}$ and \normalsize{$N_{eq}=|h_{RD}|^2|\bar{\rho }_{W_I}|^2 P_R 2^{-rq}+1$}.
\end{prop}

\textit{Sketch of the proof}: The source quantizes the interference signal
$X_{I}$ with rate $r_{q}$ using a quantization codebook with rate $I(X_I;\hat{X}_I)$. The quantization codebook is characterized by the reverse test channel $X_I=\hat{X}_I+Q$, with $Q$ being a zero-mean complex Gaussian variable with variance $P_I2^{-r_q}$, independent of $X_I$, or equivalently by the test channel $\hat{X}_I=\rho X_I+Q'$, with $\rho=1-2^{-r_q}$ and $Q^{\prime}$ being a complex Gaussian random variable with zero mean and variance $P_I2^{-r_q}(1-2^{-r_q})$, independent of $X_I$. The source inputs $X_{SD}$ and $X_{SR}$ are allocated power $\gamma P_{S}$ and
$(1-\gamma)P_{S}$, respectively where $0\leq\gamma\leq1$. We assume $X_{SD}=V$ so that the source does not forward the quantized interference $\hat{X}_I$. 
 We set $X_{SD}=\rho_{W^{\prime}}\sqrt{\gamma P_S} U_{W^{\prime }}+\rho_{W^{\prime \prime}}\sqrt{\gamma P_S} U_{W^{\prime \prime}}+\rho_{W_{I}}\sqrt{\gamma P_S} U_{W_I}$, $X_{R}=\bar{\rho}_{W^{\prime}}\sqrt{P_R} U_{W^{\prime}} + k\hat{X}_I$ and $U= U_{W^{\prime}} $, $k=\frac{|\bar{\rho }_{W_I}|\sqrt{P_{R}}}{\sqrt{\rho P_I}}$ where $U_{W^{\prime}}$, $U_{W^{\prime \prime}}$, $U_{W_I}$ are independent, zero mean, unit variance, complex Gaussian random variables and carry the messages $W^{\prime}$, $W^{\prime \prime}$ and $W_I$, respectively. Furthermore, $U_{W^{\prime}}$ and $U_{W^{\prime \prime}}$ are independent of $X_I$ and $\hat{X}_I$ whereas $\mathbb{E}[U_{W_I}X_I ]=\sqrt{P_I}$. The destination decodes messages $W^{\prime}$, $W^{\prime\prime}$ and
$W_{I}$ jointly. \hfill$\Box$


\begin{rem}
Similar to Remark \ref{rem_CU_special_case}, when we set $r_q=0$ in (\ref{(2,C,S,GC)}), Scheme (C,S,1) boils down to the special case in which the source mitigates the interference without the help of the relay using the structured approach and the relay is used for only source message cooperation.
\end{rem}


Now, we turn to scheme (C,S,2).

\begin{prop}
\label{prop(2,C,S,DM)2} For Scheme (C,S,2), the following rate is achievable for
the DM model:
\vspace{-0.4cm}
\begin{equation}
R_{(C,S,2)}=\max\min\left\{
\begin{array}
[c]{l}%
I(X_{SD};Y_{D}\hat{X}_I|X_R X_{I}U)+(I(X_{SR};Y_{R}|X_{R})-I(X_I;\hat{X}_I|UY_D))^{+},\\
(I(X_{SD}X_{I};Y_{D}\hat{X}_I|X_R U)-R_I)^+ + (I(X_{SR};Y_{R}|X_{R})-I(X_I;\hat{X}_I|UY_D))^{+},\\
I(X_{SD}X_{R};Y_{D}\hat{X}_I|X_{I} U),\\
(I(X_{SD}X_{R}X_{I};Y_{D}\hat{X}_I|U)-R_{I})^{+}%
\end{array}
\right.  ,\label{(2,C,S,DM)_2}%
\end{equation}
\vspace{-1.2cm} \newline
where the maximum is over all input pmfs $P_{U\hat{X}_{I}X_{R} X_{SR}%
X_{SD}|X_{I}}$ of the form $P_{\hat{X}_{I}|X_{I}} P_{X_{SR}|X_{R}\hat{X}_I}P_{U} P_{X_{R}%
|U}$ \newline $P_{X_{SD}|UX_{R} X_{I}}$ such that the inequality $ I(U;Y_D) \geq I(X_I;\hat{X}_I|UY_D)$ holds.
\end{prop}

\textit{Sketch of the proof}: The source quantizes the interference signal $X_I^n$ into a reconstruction sequence $\hat{X}_I^n$ by using a test channel $P_{\hat{X}_{I}|X_{I}}$. Moreover, random binning is performed according to the Wyner-Ziv strategy (See, e.g., \cite{ElGamal2}), reducing the rate of the compression codebook to $I(X_I;\hat{X}_I|UY_D)$. The source sends the index of the quantized interference and message $W^{\prime}$ to the relay. The relay recovers the compression index (but not $\hat{X}_I^n$) and $W^{\prime}$ successfully if $R^{\prime}+I(X_I;\hat{X}_I|UY_D)\leq I(X_{SR};Y_{R}|X_{R})$. The relay then maps the index of the quantized interference received from the source into a codeword $U^n$ from an independent codebook and forwards it along with the codeword that encodes message $W^{\prime}$ to the destination. The destination first decodes the codeword $U^n$, which is guaranteed if $I(U;Y_D) \geq I(X_I;\hat{X}_I|UY_D)$. From the compression index, the destination can now recover $\hat{X}_I^n$ via Wyner-Ziv decoding, since it has the side information $Y_D^n$ and $U^n$. The decoded sequence $\hat{X}_I^n$ is then used to facilitate decoding at the destination. The resulting channel to the destination is thus a MAC with common messages as (C,S,1) in which the source and the relay have the message sets ($W^{\prime}$,$W^{\prime\prime}$,$W_{I}$) and
$W^{\prime}$, respectively. Unlike (C,S,1), here the destination has the knowledge of both
$\hat{X}_{I}^{n}$ and $U^n$, which is used to jointly decode the messages set ($W^{\prime}$,$W^{\prime\prime}$,$W_{I}$). Similar to (C,S,1), an achievable rate region is obtained as \vspace{-0.4cm}
\begin{subequations}
\label{CS_MAC2}%
\begin{align}
R^{\prime\prime} &  \leq I(X_{SD};Y_{D}\hat{X}_I|X_R X_{I} U)\\
R^{\prime\prime}+R_{I} &  \leq I(X_{SD} X_{I};Y_{D}\hat{X}_I|X_R  U)\\
R^{\prime\prime}+R^{\prime} &  \leq I(X_{SD}X_R;Y_{D}\hat{X}_I|X_{I} U)\\
R^{\prime\prime}+R^{\prime}+R_{I} &  \leq I(X_{SD} X_R X_{I};Y_{D}\hat{X}_I|U),
\end{align}
\vspace{-1.3cm} \newline
for some input pmf $P_{UX_{R}X_{SD}|X_{I}}= P_{U} P_{X_{R}|U }P_{X_{SD}|UX_R X_{I}}$. Incorporating the constraints on $R^{\prime}$ and $I(U;Y_D)$ above into
(\ref{CS_MAC2}) gives us (\ref{(2,C,S,DM)_2}). \hfill$\Box$
\end{subequations}
\begin{prop}
For Scheme (C,S,2), the following rate is achievable for the AWGN model
(\ref{gaussian_model}):
\begin{align}
R_{(C,S,2)}= &\displaystyle\max_{\substack{r_{q},\rho_{W^{\prime}},\rho
_{W^{\prime\prime}},\\ \rho_{W_{I}}, \rho_{U}, \bar{\rho}_{W^{\prime}},\bar{\rho}_{U},\gamma:\\|\rho_{W^{\prime}}|,|\rho_{W^{\prime\prime}}|,|\rho_{W_{I}%
}|,\\|\rho_{U}|,|\bar{\rho}_{W^{\prime}}|,|\bar{\rho}_{U}|,\gamma\in
\lbrack0,1]}}   \min\left\{
\begin{array}
[c]{ll}%
\mathcal{C}\left(  P_{W^{\prime\prime}}\right)  +(\mathcal{C}\left(
|h_{SR}|^{2}(1-\gamma)P_{S}\right)  -r_{q})^{+}, & \\
(\log_2\left((1+P_{W^{\prime\prime}})\frac{P_I}{D}+P_{W_{I}}\right)-R_I)^+ & \\ +(\mathcal{C}\left(
|h_{SR}|^{2}(1-\gamma)P_{S}\right)  -r_{q})^{+}, & \\
\mathcal{C}\left(  P_{W^{\prime\prime}}+P_{W^{\prime}}\right)  , & \\
(\log_2\left((1+P_{W^{\prime\prime}}+P_{W^{\prime}})\frac{ P_I}{D}+P_{W_{I}}\right)-R_I)^+ &
\end{array}
\right.  \label{(2,C,S,GC)2}\\
&\hbox{ subject to }   0\leq r_{q}\leq \min \left\{
                                              \begin{array}{ll}
                                                \mathcal{C}\left(  |h_{SR}|^{2}%
(1-\gamma)P_{S}\right), & \hbox{} \\
                                                \mathcal{C}\left(\frac{P_U}{P_{W^{\prime}}+P_{W^{\prime\prime}}+P_{W_I}+1}\right)  & \hbox{}
                                              \end{array}
                                            \right.
 \nonumber
\\ &  |\rho_{W^{\prime}}|^{2}+|\rho_{W^{\prime\prime}}|^{2}+|\rho_{W_{I}}%
|^{2} + |\rho_{U}|^{2} \leq1 \hbox{ and } |\bar{\rho}_{W^{\prime}}|^{2}+|\bar{\rho}_{U}|^{2}\leq1\nonumber
\end{align}
\vspace{-1.2cm}\newline
where $P_{W^{\prime}}=(|h_{RD}||\bar{\rho}_{W^{\prime}}|\sqrt{P_{R}}%
+|h_{SD}||\rho_{W^{\prime}}|\sqrt{\gamma P_{S}})^{2}$, $P_{W^{\prime
\prime}}=|h_{SD}|^{2}|\rho_{W^{\prime\prime}}|^{2}\gamma P_{S}$, $P_{W_{I}}=(|h_{SD}||\rho_{W_{I}}|\sqrt{\gamma P_{S}}+|h_{I}|\sqrt{P_{I}})^{2}$, $P_{U}=(|h_{SD}||\rho_{U}|\sqrt{\gamma P_{S}}+|h_{RD}||\bar{\rho}_{U}|\sqrt{P_{R}})^{2}$, $D=P_I 2^{-r_q}\frac{(1-x)}{1-x2^{-r_q}}$ and $x=P_{W_I}/(P_{W^{\prime}}+P_{W^{\prime\prime}}+P_{W_I}+1)$.
\end{prop}

\textit{Sketch of the proof}: Similar to (C,S,1), the source quantizes the interference signal
$X_{I}$ with rate after binning, given by $r_{q}=I(X_I;\hat{X}_I|UY_D)$. The quantization codebook is characterized by the reverse test channel $X_I=\hat{X}_I+Q$, with $Q$ being a zero-mean complex Gaussian variable with variance $D$, independent of $X_I$, or equivalently the test channel $\hat{X}_I=\rho X_I+Q'$, with $\rho=1-D/P_I$ and $Q^{\prime}$ being a complex Gaussian random variable with zero mean and variance $D(1-D/P_I)$, independent of $X_I$. We obtain $D=P_I 2^{-r_q}\frac{(1-x)}{1-x2^{-r_q}}$ where $x$ is defined above. The term $\frac{(1-x)}{1-x2^{-r_q}}$ represents the percentage of the decreased distortion due to side information about $X_I$ at the destination. When $x=0$, $D=P_I2^{-r_q}$ which is the case where there is no side information about $X_I$ at the destination. As $x\rightarrow 1$, $D\rightarrow 0$ for any nonzero $r_q$ and the destination can completely recover $X_I$ using the side information. The source inputs $X_{SD}$ and $X_{SR}$ are allocated power $\gamma P_{S}$ and $(1-\gamma)P_{S}$, respectively where $0\leq\gamma\leq1$. 
 We set $X_{SD}=\rho_{W^{\prime}}\sqrt{\gamma P_S} U_{W^{\prime }}+\rho_{W^{\prime \prime}}\sqrt{\gamma P_S} U_{W^{\prime \prime}}+\rho_{W_{I}}\sqrt{\gamma P_S} U_{W_I}+ \rho_{U}\sqrt{\gamma P_S} U$ and $X_{R}=\bar{\rho}_{W^{\prime}}\sqrt{P_R} U_{W^{\prime}} + \bar{\rho}_{U}\sqrt{P_R}U $ where $U_{W^{\prime}}$, $U_{W^{\prime \prime}}$, $U_{W_I}$ and $U$ are independent, zero mean, unit variance, complex Gaussian random variables and carry the messages $W^{\prime}$, $W^{\prime \prime}$, $W_I$ and the index of the compressed interference, respectively. Furthermore, $U_{W^{\prime}}$, $U_{W^{\prime \prime}}$ and $U$ are independent of $X_I$ and $\hat{X}_I$ whereas $\mathbb{E}[U_{W_I}X_I ]=\sqrt{P_I}$. The destination first decodes the codeword $U$ and thus recovers $\hat{X}_I$, and then it decodes messages $W^{\prime}$, $W^{\prime\prime}$ and $W_{I}$ jointly using the knowledge of $U$ and $\hat{X}_I$. \hfill$\Box$
\vspace{-0.5cm}
\subsection{Discussions}
For comparison purposes, we also show the performance
of the Scheme Analog Input Description, referred to as AID \cite{Zaidi_new} \cite{Zaidi_orthogonal}. In this scheme, the source generates the codeword to be
transmitted by the relay as if the relay knew the interference and the message
non-causally and they used DPC jointly. The source then quantizes this
codeword and sends it to the relay through the source-relay link. The relay
simply forwards a scaled version of the quantized signal received from the
source. The achievable rate for DM and AWGN are given in [17, Theorem 2] and [17, Theorem 4], respectively. For the DMC model, [17, Theorem 2] can
be easily modified by setting $V=X_{1R}$. For Gaussian case, we incorporate
complex channel gains into [1, Theorem 4] and obtain \vspace{-0.2cm}%
\begin{align}
R_{AID}= &  \displaystyle\max_{\substack{\gamma:\gamma\in\lbrack
0,1]}}\mathcal{C}\left(  \frac{(|h_{SD}|\sqrt{\gamma P_{S}}+|h_{RD}%
|\sqrt{P_{R}-D})^{2}}{1+|h_{RD}|^{2}D}\right)  \\
&  \hbox{ where }D=\frac{P_{R}}{|h_{SR}|^{2}(1-\gamma)P_{S}+1}.\nonumber
\end{align} \vspace{-1.2cm}\newline
A special case of the model presented in this paper is a multihop channel characterized by\newline
$P_{Y_{R}Y_{D}|X_{SD},X_{SR},X_{R},X_{I}}=P_{Y_{R}|X_{SR}}P_{Y_{D}|X_{R},X_{I}}$.
The achievable rates of this section can be
easily specialized to the multihop channel. Specifically, for DM model, we remove the dependence of $Y_{D}$ on
$X_{SD}$ and we set $X_{SD}={\O }$. For Gaussian case, we set
$h_{SD}=0$ and hence $X_{SD}=0$. An
achievable rate  for the multihop channel by treating the interference as i.i.d. state was derived in \cite{song}. This scheme, denoted by NL-DF,
utilizes nested lattice codes to cancel an integer part of the interference
while treating the residual of the interference as noise. The achieved rate for AWGN model can be written as \cite{song}: \vspace{-0.3cm}
\begin{equation}
R_{NL-DF}=\left[  \log\left(  \frac{|h_{SR}|^{2}|h_{RD}|^{2}P_{S}P_{R}%
+|h_{SR}|^{2}P_{S}+|h_{RD}|^{2}P_{R}+1}{|h_{SR}|^{2}P_{S}+|h_{RD}|^{2}P_{R}%
+2}\right)  \right]  ^{+}.\label{eq_song}%
\end{equation}

\section{On the Optimality of Interference Forwarding}
\label{capacity_results} In this section, we consider a special case of general model considered so far
where $Y_{D}=(Y_{D_{1}},Y_{D_{2}})$ and the channel to the destination
factorizes as \vspace{-0.6cm}
\begin{equation}
P_{Y_{D}|X_{SD},X_{R},X_{I}}=P_{Y_{D_{1}}|X_{SD}}\cdot P_{Y_{D_{2}}%
|X_{R},X_{I}},\label{special_class_prob_model}%
\end{equation} \vspace{-1.4cm}\newline
as depicted in Fig. \ref{Special_class_relay_channel_w_interference_informed_source}. This
corresponds to a model where the links $S-D$ and $R-D$ are orthogonal
to each other, in addition to being orthogonal to the $S-R$ channel $P_{Y_{R}|X_{SR},X_{R}}$. In
other words, this scenario can be seen as the parallel of a multihop channel
$S-R-D$ and a direct channel $S-D$. Moreover, from
(\ref{special_class_prob_model}), the interference affects the $R-D$ channel
only. We are interested in obtaining general guidelines on how the
interference information at the source should be leveraged. In particular,
since the interference only affects one of the parallel channels, namely the
multihop link $S-R-D,$ should the $S-D$ channel be used to provide
interference information so as to facilitate decoding on the $S-R-D$ link? A
similar question can be of course posed for the case where interference
affects only the $S-D$ link.

The question is motivated by reference \cite{1}, where it is shown that if the
interference is \textit{unstructured} and the relay is informed about the
source message (but not the interference), interference information should
\textit{not} be forwarded on the $S-D$ link. A related scenario is also
considered in \cite{Zaidi_new}, where instead \textit{unstructured}
interference affects the $S-R$ and $S-D$ links only, in a dual manner with
respect to the model at hand.

We tackle the question above first for the DMC\ model. The next proposition
shows that, even with structured interference, there is no advantage in using
the $S-D$ link for interference management.

\begin{prop}
\label{converse}In the model of Fig.
\ref{Special_class_relay_channel_w_interference_informed_source}, capacity is
achieved by transmitting independent information on the multihop link $S-R-D$
and on the $S-D$ link. Moreover, the signal sent on the $S-D$ link can be
chosen to be independent of the interference signal.
\end{prop}

\begin{proof}
We prove this result by evaluating the capacity in multiletter form and
arguing that the derived capacity can be achieved by a scheme that complies
with the statement of Proposition \ref{converse}. In particular, we prove that
the capacity is given by $C=C_{SD}+C_{SRD},$ where \vspace{-0.3cm}
\begin{align}
C_{SD} &  =\max_{P_{X_{SD}}}I(X_{SD};Y_{D_{1}})\label{converse_1}
\\ \text{ and }C_{SRD} &  =\max_{P_{X_{SR}^{n}|X_{I}^{n}},P_{X_{R}^{n}|Y_{R}^{n}%
}}I(X_{SR}^{n};Y_{D_{2}}^{n}),\label{converse_2}%
\end{align}
\vspace{-1.2cm}\newline
and $P_{X_{R}^{n}|Y_{R}^{n}}=%
{\displaystyle\prod\nolimits_{i=1}^{n}}
P_{X_{Ri}|Y_{SR}^{i-1}}.$ Note that $C_{SD}$ is the maximum rate achievable on
the (interference-free) $S-D$ link, which is given by the standard
point-to-point capacity (\ref{converse_1}), while $C_{SRD}$ is the maximum
achievable rate on the $S-R-D$ link. The latter cannot in general be
calculated as a single-letter expression, unlike $C_{SD}.$ Moreover, note that
(\ref{converse_2}) is simply achieved by using the encoding strategies
described by pmfs $P_{X_{SR}^{n}|X_{I}^{n}}$ and
$P_{X_{R}^{n}|Y_{R}^{n}}$. Since by these arguments, $C$ is achievable, we
only need to prove that $C$ is also an upper bound on the capacity. This is
done in Appendix \ref{appendix_converse}.
\end{proof}

We now specialize the result above to the corresponding Gaussian model shown
in Fig. \ref{Special_class_relay_channel_w_interference_informed_source_awgn_special},
which is described by the input and output relations at time instant $i$ \vspace{-0.5cm}
\label{gaussian_model_special}%
\begin{align}
\resizebox{.9 \hsize}{!}{$ Y_{R,i}   =h_{SR,i}X_{SR,i}+Z_{R,i}\hbox{, }Y_{D_{1},i}    =h_{SD,i}X_{SD,i}+Z_{D_{1},i} \hbox{ } \hbox{ } \hbox{ and }\hbox{ }\hbox{ } Y_{D_{2},i}   =h_{RD,i}X_{R,i}+h_{I,i}X_{I,i}+Z_{D_{2},i}$}%
\end{align} \vspace{-1.3cm}\newline
where the noises $Z_{D_{1},i}$ and $Z_{D_{2},i}$ are independent zero mean
complex Gaussian random variable with unit variance. The result of Proposition
\ref{converse} can be easily generalized to a scenario with power constraints
and can thus be applied also to the Gaussian model. Specifically, to simplify
our results, we impose two separate power constraints on $X_{SR}$ and $X_{SD}$ as
\vspace{-0.2cm}
\label{power_constraints}%
\begin{equation}
\frac{1}{n}\sum_{i=1}^{n}\mathbb{E}\left[  |X_{SR,i}|^{2}\right]  \leq
P_{SR}\hbox{ }\hbox{ }\text{ and }\hbox{ }\hbox{ }\frac{1}{n}\sum_{i=1}^{n}\mathbb{E}\left[  |X_{SD,i}%
|^{2}\right]  \leq P_{SD},
\end{equation} \vspace{-1.1cm} \newline
along with the relay power constraint in (\ref{relay_power_constraint}). The
following Proposition obtains the capacity for this model in a more explicit
way than (\ref{converse_1})-(\ref{converse_2}) for some special cases. Note
that $C_{SD}=\mathcal{C}\left(  |h_{SD}|^{2}P_{SD}\right)  ,$ while $C_{SRD}$
is generally unknown. We define%
\vspace{-0.3cm}
\begin{equation}
C_{SRD}^{\prime}=\max\left\{
\begin{array}
[c]{l}%
\mathcal{C}\left(  \frac{|h_{RD}|^{2}P_{R}}{1+|h_{I}|^{2}P_{I}}\right)  ,\\
\min\left\{\mathcal{C}\left(  |h_{RD}|^{2}P_{R}\right)  ,(\mathcal{C}\left(  |h_{RD}|^{2}P_{R}+|h_{I}|^{2}P_{I}\right)  -R_{I})^{+} \right\}
\end{array}
\right.  .
\end{equation}

\begin{prop}
If $\mathcal{C}(|h_{SR}|^{2}P_{SR})\geq R_{I}+\mathcal{C}\left(  |h_{RD}%
|^{2}P_{R}\right)  $, then the scheme (D,U) is optimal and the capacity is
given by \vspace{-0.2cm}
\begin{equation}
C=\mathcal{C}\left(  |h_{SD}|^{2}P_{SD}\right)  +\mathcal{C}\left(
|h_{RD}|^{2}P_{R}\right)  . \label{cap_res_1}%
\end{equation} \vspace{-1.2cm}\newline
If instead $\mathcal{C}(|h_{SR}|^{2}P_{SR})\leq C_{SRD}^{\prime},$ then a
scheme that chooses the best strategy between (N,S) and (N,U) for the given
system parameters is optimal and the capacity is $C=\mathcal{C}\left(  |h_{SD}|^{2}P_{SD}\right)  +\mathcal{C}\left(
|h_{SR}|^{2}P_{SR}\right) $.
\end{prop}

\begin{proof}
If $\mathcal{C}(|h_{SR}|^{2}P_{SR})\geq R_{I}+\mathcal{C}\left(  |h_{RD}%
|^{2}P_{R}\right)  ,$ then the source can provide both interference and useful
message to the relay without loss of optimality, since the rate of the
message can never be larger than $\mathcal{C}\left(  |h_{RD}|^{2}P_{R}\right)
$ by cut-set arguments. Scheme (D,U) is thus optimal and achieves the
interference-free capacity (\ref{cap_res_1}). The case $\mathcal{C}(|h_{SR}|^{2}P_{SR})\leq
C_{SRD}^{\prime}$ is more complex. From \cite{osvaldo} \cite{Baccelli} it is
known that the maximum rate on the R-D link, assuming that the relay is
unaware of the interference is given by $C_{SRD}^{\prime}.$ This is achieved
by having the destination either treat interference as noise or perform
joint decoding of source information and interference. By the cut-set bound we also know that $C_{SRD}\leq
\mathcal{C}(|h_{SR}|^{2}P_{SR}).$ However, rate $C_{SRD}=\mathcal{C}%
(|h_{SR}|^{2}P_{SR})$ is achievable if $\mathcal{C}(|h_{SR}|^{2}P_{SR})\leq
C_{SRD}^{\prime}$ by not informing the destination about the interference and using the
decoding strategy that attains $C_{SRD}^{\prime}.$
\end{proof}
\vspace{-0.5cm}
\section{Ergodic Fading}

\label{ergodic fading} In this section, we study the effect of ergodic fading in model (\ref{gaussian_model})
on the performance of the proposed schemes.
We recall that the instantaneous values of the channels are only known to the receivers, while the transmitters
only have knowledge of the channel statistics. As for the latter, we assume
that channel gains $h_{SR}$, $h_{SD}$, $h_{RD}$ and $h_{I}$ are independent
Ricean distributed with parameters $K_{SR}$, $K_{SD}$, $K_{RD}$, $K_{I}$,
i.e., $h_{SR}=\mu_{SR}+z_{SR}$ where $\mu_{SR}$ represents the direct
(deterministic) line of sight component and $z_{SR}\sim \mathcal{CN}(0,\sigma_{SR}^{2})$ such that
$|\mu_{SR}|^{2}+\sigma_{SR}^{2}=1$ and $|\mu_{SR}|^{2}/\sigma_{SR}^{2}=K_{SR}%
$, and likewise for other channel gains.
\vspace{-0.7cm}
\subsection{No Relay Case} \label{Sec_no_relay_ergodic}
We first study the point-to-point channel, i.e., where the relay
is not present. This forms a foundation of the multihop relay channel investigated in Sec. \ref{section_fading_multihop}. For this scenario, the achievable rate with the unstructured approach is given by
\vspace{-0.2cm} \begin{align}
R_{U}=\max_{\alpha}\mathbb{E}\left[  \log_{2}\left(  \frac{(|h_{SD}|^{2}%
P_{S})(|h_{SD}|^{2}P_{S}+|h_{I}|^{2}P_{I}+1)}{|h_{SD}|^{2}|h_{I}|^{2}%
P_{S}P_{I}(1-2Re(\alpha)+|\alpha|^{2})+\alpha^{2}|h_{I}|^{2}P_{I}+|h_{SD}%
|^{2}P_{S}}\right)  \right]  .\label{DPC_fading_Rate}%
\end{align}
\vspace{-1.0cm}\newline
We employ GP coding with linear assignment of auxiliary random variable $U$ with an inflation factor $\alpha$ \cite{Bennatan}. The parameter $\alpha$ is chosen to be fixed for all fading levels due to the lack of CSIT and is optimized numerically, as opposed to the approaches in \cite{Zhang}, \cite{Vaze} and \cite{mitran}.

For the structured approach, from \cite{maric}, we easily obtain the achievable rate%
\vspace{-0.3cm} \begin{align}\label{no_relay_fading_rate_structured}
R_{S}=\max_{\substack{\rho,\rho_{I},\rho_{I}^{\prime}:\\|\rho|,|\rho
_{I}|,|\rho_{I}^{\prime}|\in\lbrack0,1]}} &  \min\left\{
\begin{array}
[c]{ll}%
\mathbb{E}\left[  \mathcal{C}\left(  |h_{SD}\rho|^{2}P_{S}\right)  \right]
), & \\
\left(  \mathbb{E}\left[  \mathcal{C}\left(  |h_{SD}\rho|^{2}P_{S}+|h_{SD}%
\rho_{I^{\prime}}|^{2}P_{S}+|h_{SD}\rho_{I}\sqrt{P_{S}}+h_{I}\sqrt{P_{I}}%
|^{2}\right)  \right]  -R_{I}\right)  ^{+} &
\end{array}
\right.  \nonumber\\
&  \hbox{ subject to }|\rho|^{2}+|\rho_{I}|^{2}+|\rho_{I^{\prime}}|^{2}\leq1
\end{align} \vspace{-1.3cm}\newline
where the source allocates powers for forwarding its own message and interference to the destination. In particular, power $|\rho_{I}|^{2}P_{S}$ is
used to transmit interference by forwarding the same codeword transmitted by
the interferer, while power $|\rho_{I^{\prime}}|^{2}P_{S}$ is devoted to
transmission of the interference message via an independently generated
codeword. The rationale for this is that, as $K\rightarrow\infty$,
fading becomes deterministic and it is optimal for the source to transmit
coherently with the interferer by setting $\rho_{I^{\prime}}=0$. Instead, as
$K\rightarrow0$ (Rayleigh fading), it is more advantageous for the source to
forward interference by using an independent codebook by setting $\rho_{I}=0$.
Hence, the source employs both of the interference forwarding strategies to
accommodate intermediate $K$ values.
\vspace{-0.5cm}
\subsection{Multihop Relay Channel}\label{section_fading_multihop}

In this section, we include the relay in the ergodic fading model by
considering the special case of a multihop relay channel, i.e., $h_{SD}=0$. The detailed analysis and insights can be extended to the general orthogonal components relay channel.

The following propositions report the achievable rates of the proposed schemes
for the scenario at hand. The proofs are straightforward consequences of the
analysis above and Sec. \ref{ach_rates}.
%
%
%
\begin{prop}
For (D,U), the following rate is achievable for the multihop fading model:
\vspace{-0.3cm}\begin{equation}
R_{(D,U)}=\max_{\alpha}\min\left\{
\begin{array}
[c]{l}%
\left(  \mathbb{E}\left[  \mathcal{C}\left(  |h_{SR}|^{2}P_{S}\right)
\right]  -R_{I}\right)  ^{+},\\
\mathbb{E}\left[  \log_{2}\left(  \frac{(|h_{RD}|^{2}P_{R})(|h_{RD}|^{2}%
P_{R}+(|h_{I}|^{2}P_{I}+1)}{|h_{RD}|^{2}|h_{I}|^{2}P_{R}P_{I}(1-2Re(\alpha
)+|\alpha|^{2})+\alpha^{2}|h_{I}|^{2}P_{I}+|h_{RD}|^{2}P_{R}}\right)  \right]
\end{array}
\right.
\end{equation} \vspace{-0.8cm}\newline
\end{prop}
\vspace{-0.6cm}
\begin{prop}
For (D,S), the following rate is achievable for the multihop fading model:
\vspace{-0.6cm}\begin{align}
R_{(D,S)}=\max_{\substack{\bar{\rho},\bar{\rho}_{I},\bar{\rho}_{I^{\prime}}:\\|\bar{\rho}|,|\bar{\rho}
_{I}|,|\bar{\rho}_{I^{\prime}}|\in\lbrack0,1]}} &  \min\left\{
\begin{array}
[c]{ll}%
\left(  \mathbb{E}\left[  \mathcal{C}\left(  |h_{SR}|^{2}P_{S}\right)
\right]  -R_{I}\right)  ^{+}, & \\
\mathbb{E}\left[  \mathcal{C}\left(  |h_{RD}\bar{\rho}|^{2}P_{R}\right)  \right]
), & \\
\left(  \mathbb{E}\left[  \mathcal{C}\left(  |h_{RD}\bar{\rho}|^{2}P_{R}+|h_{RD}%
\bar{\rho}_{I^{\prime}}|^{2}P_{R}+|h_{RD}\bar{\rho}_{I}\sqrt{P_{R}}+h_{I}\sqrt{P_{I}}%
|^{2}\right)  \right]  -R_{I}\right)  ^{+} &
\end{array}
\right. \nonumber \\
&  \hbox{ subject to }|\bar{\rho}|^{2}+|\bar{\rho}_{I}|^{2}+|\bar{\rho}_{I^{\prime}}|^{2}\leq1
\end{align} \vspace{-2cm}\newline
\end{prop}

\begin{rem} For Gaussian model (\ref{gaussian_model}), structured strategies in no-relay case as well as in (D,S) are inferior to the unstructured ones in no-fading case due to the ability of DPC completely eliminating the effect of the interference. However as shown in Section \ref{num_res}, these strategies become meaningful under fading where precoding can not completely cancel the interference.
\end{rem}
\begin{prop}
For (C,U), the following rate is achievable for the multihop fading model: \vspace{-0.5cm}
\begin{align}
 R_{(C,U)}= \max_{r_q,\alpha} &\min\left\{
\begin{array}
[c]{l}%
\left(  \mathbb{E}\left[  \mathcal{C}\left(  |h_{SR}|^{2}P_{S}\right)
\right]  -r_q\right)  ^{+},\\
\mathbb{E}\left[  \log_{2}\left(  \frac{(|h_{RD}|^{2}P_{R})(|h_{RD}|^{2}%
P_{R}+(|h_{I}|^{2}(P_{I}-D)+N_{eq})}{|h_{RD}|^{2}|h_{I}|^{2}P_{R}(P_{I}-D)(1-2Re(\alpha
)+|\alpha|^{2})+N_{eq}(\alpha^{2}|h_{I}|^{2}(P_{I}-D)+|h_{RD}|^{2}P_{R})}\right)  \right]
\end{array}
\right.
\end{align} \vspace{-1.2cm} \newline
where $N_{eq}=|h_I|^2D+1$ and $D=P_I2^{-r_q}$.
\end{prop}

\begin{prop}
For (C,S,1), the following rate is achievable for the multihop fading model:%
\vspace{-0.3cm}\begin{align}
 R_{(C,S,1)}= \max_{\substack{r_{q},\bar{\rho},\bar{\rho}_{I},\bar{\rho}_{I^{\prime}}:\\|\bar{\rho}
|,|\bar{\rho}_{I}|,|\bar{\rho}_{I^{\prime}}|\in\lbrack0,1]}} & \min\left\{
\begin{array}
[c]{ll}%
\left(  \mathbb{E}\left[  \mathcal{C}\left(  |h_{SR}|^{2}P_{S}\right)
\right]  -r_{q}\right)  ^{+}, & \\
\mathbb{E}\left[  \mathcal{C}\left(|h_{RD}\bar{\rho}|^{2}P_{R}%
/N_{eq} \right)  \right] , & \\
(  \mathbb{E}[  \mathcal{C}(( |h_{RD}\bar{\rho}|^{2}P_{R} + |h_{RD}\bar{\rho }_{I^{\prime}}|^{2} P_{R}(1-2^{-r_{q}})+ & \\ |h_{RD}\bar{\rho}_I  \sqrt{P_{R}(1-2^{-r_{q}})}+h_{I}\sqrt{P_{I}}|^{2})/N_{eq} )  ] -R_{I})  ^{+} &
\end{array}
\right.\\  & \hbox{ subject to }|\bar{\rho}|^{2}+|\bar{\rho}_{I}|^{2}+|\bar{\rho}_{I^{\prime}}|^{2}\leq1 \nonumber
\end{align}  \vspace{-1.2cm}\newline
where $N_{eq}=|h_{RD}\bar{\rho}_I|^2P_R 2^{-r_q}+|h_{RD}\bar{\rho}_{I^{\prime}}|^2P_R 2^{-r_q}+1$.
\end{prop}

\begin{prop}
For (C,S,2), the following rate is achievable for the multihop fading model:
 \vspace{-0.4cm}\begin{eqnarray}
&R_{(C,S,2)}=\displaystyle \max_{\substack{r_{q},\bar{\rho},\bar{\rho}_{U}:\\|\bar{\rho}
|,|\bar{\rho}_{U}|\in\lbrack0,1]}}\min\left\{
\begin{array}
[c]{ll}%
\left(  \mathbb{E}\left[  \mathcal{C}\left(  |h_{SR}|^{2}P_{S}\right)
\right]  -r_{q}\right)  ^{+}, & \\
\mathbb{E}\left[  \mathcal{C}\left(|h_{RD}\bar{\rho}|^{2}P_{R}%
 \right)  \right] , & \\
\left(  \mathbb{E}\left[  \log_2\left( (|h_{RD}\bar{\rho}|^{2}P_{R}+1)2^{r_q} +|h_{I}|^2 P_{I}  \right)  \right]
-R_{I}\right)  ^{+} &
\end{array}
\right. \\  & \hbox{ subject to }|\bar{\rho}|^{2}+|\bar{\rho}_{U}|^{2}\leq1 \hbox{ and } r_q \leq \mathbb{E}\left[  \mathcal{C}\left(\frac{|h_{RD}\bar{\rho}_{U}|^{2}P_{R}}{|h_{RD}\bar{\rho}|^{2}P_{R} +|h_{I}|^2 P_{I} + 1 } \right)  \right] \nonumber
\end{eqnarray} \vspace{-1.8cm}\newline
\end{prop}
\begin{rem} In the fading scenario for Scheme (C,S,2), the source does not know $h_{RD}$ and thus can not determine the instantaneous Wyner-Ziv compression rate to compress $X_I$ with respect to the destination observation. Therefore, for simplicity, we assume that the source neglects the side information available at the destination and does not perform binning. Recall that neglecting the side information corresponds to the case where $x=0$ in (\ref{(2,C,S,GC)2}). \end{rem}

\begin{prop}
For AID, the following rate is achievable for the multihop fading model:
\vspace{-0.4cm}\begin{equation}
R_{AID}=\resizebox{.8 \hsize}{!}{$\max_{\alpha}\mathbb{E}\left[  \log_{2}\left(  \frac{(|h_{RD}|^{2}%
(P_{R}-D))(|h_{RD}|^{2}(P_{R}-D)+|h_{I}|^{2}P_{I}+N_{eq})}{|h_{RD}|^{2}|h_{I}|^{2}%
(P_{R}-D)P_{I}(1-2Re(\alpha)+|\alpha|^{2})+N_{eq}(\alpha^{2}|h_{I}|^{2}P_{I}+|h_{RD}%
|^{2}(P_{R}-D))}  \right) \right]$}
\end{equation} \vspace{-1.2cm}\newline
where $N_{eq}=|h_{RD}|^{2}D+1$, $D=P_R 2^{-r_q}$ and $r_q=\mathbb{E}\left[  \mathcal{C}\left(  |h_{SR}|^{2}P_{S}\right)
\right] $. The source evaluates the signal to be transmitted by the relay when the relay utilizes the unstructured approach, namely DPC for $(R-D)$ ergodic channel. The source quantizes the corresponding signal with rate $r_q$ and forwards it to the relay. The relay simply forwards the received signal to destination.
\end{prop}
\vspace{-0.4cm}
\section{Numerical Results}

\label{num_res} In this section, we numerically evaluate the achievable rates
for the AWGN models, both with no fading and with ergodic fading, and compare
them with two following simple schemes.

\begin{itemize}
\item \textit{Scheme No Relay (NR)}: The achieved rate is given by \cite{Costa}
and denoted as $R_{NR}$;

\item \textit{Scheme No Interference (NI)}: We set $P_{I}=0$ and $R_{I}=0,$ so
that the interference is not present. The capacity for this scenario, $R_{NI}$, is
achieved by PDF \cite{ElGamal} and is given by (\ref{(2,D,U,GC)}) with
$R_{I}=0$. Note that $R_{NI}$ provides an upper
bound to rates of the proposed achievable schemes.
\end{itemize}
\vspace{-0.6cm}
\subsection{No Fading}

We first consider the no fading case. In Fig. \ref{fig_1}, the achievable rates are illustrated as a function of the
interference power $P_{I}$ for $P_{S}=P_{R}=10dB$, $|h_{SD}|=|h_{SR}%
|=|h_{RD}|=|h_{I}|=1$ and $R_{I}=1$ bits/channel use. Scheme (C,U) outperforms
all others for low interference power, since in this case cooperative
interference mitigation strategies are not worth the capacity needed on the
source-relay link for digital interference sharing. Moreover, leveraging the
interference structure is not useful since interference decoding at the
destination is hindered by the low interference power. For large $P_{I}$,
Scheme (C,S,2) instead outperforms all others and eventually meets the upper bound $R_{NI}%
$. The larger $P_{I}$ is, the less power the source and the relay need to make the
interference decodable at the destination. In fact if $P_{I}$ is sufficiently
large, the destination is able to decode the interference without the help of
the source or the relay and the schemes which utilize structured approach, namely (C,S,1) and (C,S,2) achieve interference-free bound and hence they are optimal. We also note that as the interference power increases, Schemes (C,S,1) and (C,S,2) perform the same and have $r_q=0$ which means that the relay is utilized only for forwarding the source message. Scheme (D,U) completely
eliminates the interference by MU-DPC when $R_{I}$ is greater than the
capacity of the source-relay link, as is the case here, and hence, the
performance of Scheme (D,U) is independent of the interference power. However,
there is a gap between the performance of Scheme NI and Scheme (D,U) due to
the source-relay capacity used for informing the relay about the interference.
Similarly the performance of the scheme (AID) also does not depend on the interference power.

In Fig. \ref{fig_2}, we set the source-relay channel gain to $|h_{SR}|=2$
and keep the rest of the parameters same as Fig. \ref{fig_1} in order to study the effects of a higher gain for source-relay channel. We observe that Scheme (D,U) outperforms all
schemes for moderate interference power $P_{I}$. Now the source
and the relay are able to better mitigate the interference jointly since the (S-R) channel has enough capacity for conveying digital interference information to the relay. In fact for large $|h_{SR}|$, the capacity of source-relay channel is high enough to share the interference with the relay digitally at no extra cost and Scheme (D,U) achieves the interference-free upper bound. In Fig. \ref{fig_3}, we increase the
interference rate and set $R_{I}=3$ bits/channel use by keeping the rest of
the parameters the same as for Fig. \ref{fig_2}. We observe that (AID) outperforms (D,U) as well as all other schemes for moderate
interference power. Since the interference rate is large compared to the source relay channel capacity, informing the relay about the
interference in a digital fashion becomes too costly.


Finally, we illustrate the achievable rate as a function of $R_{I}$ in Fig.
\ref{fig_5} for the multihop relay channel $|h_{SD}|=0$ and we set $P_{S}=P_{R}=P_{I}=10dB$  and $|h_{SR}|=|h_{RD}%
|=|h_{I}|=1$. We also include Scheme NL-DF whose performance is independent of $R_I$. For small interference rate, schemes that exploit the interference structure at the destination, namely (C,S,1) and (C,S,2), result in the best rate and achieve no-interference upper bound. As the interference rate increases, schemes (D,U), (C,S,1) and (C,S,2) degrade in performance since it is harder to decode the interference at either the relay or the destination. Note also that for moderate interference rates,
Scheme (C,S,2) outperforms all others, showing that interference sharing via
compressed information along with a structured approach is the most beneficial
strategy in this regime.
\vspace{-0.4cm}
\subsection{Ergodic Fading}
\label{Sec_num_ef}
In this section, we turn to fading channels. We first consider the
point-to-point case, i.e., $h_{SR}=h_{RD}=0$. In Fig. \ref{fading_fig_2}, we
illustrate the rate as a function of the interference power for $P_{S}=5dB$
and Ricean factor $K=1$ for both $h_{SD}$ and $h_{I}$ channel gains. As the
interference power increases, the structured approach outperforms the
unstructured one. Recall that, in the case of no fading unstructured approach, namely DPC, achieves the no-interference upper bound
and hence is optimal. However, for fading channels with no channel knowledge at the source, the unstructured approach is not able to
completely cancel the interference anymore, and the structured approach becomes beneficial when the
interference power is large. To get further insights on this, in Fig.
\ref{fading_fig_1}, the rate as a function of parameter $K$, common for $h_{SD}$ and $h_I$, is illustrated for
various interference rates when $P_{S}=P_{I}=5$dB. We observe that as $K$ increases,
the gap between the no-interference upper bound and the performance of the
unstructured approach decreases and, as $K\rightarrow\infty$, the unstructured
approach achieves the no-interference bound. This is expected since, as
$K\rightarrow\infty,$ the channel model becomes equivalent to the no-fading case. For
small $K$, instead, the structured approach outperforms the structured
approach for small $R_{I}$.

Finally, we study multihop relay channel where $h_{SD}=0$ and $h_{SR}$, $h_{RD}$
and $h_{I}$ are Ricean distributed with the same parameter $K$. In Fig.
\ref{fading_fig_3}, the rate as a function of interference power is
illustrated when $P_{S}=10$dB, $P_{R}=7$dB, $R_{I}=0.4$ bits/channel use and
$K=1$. We do not include Scheme (C,S,1) in Fig. \ref{fading_fig_3} since it is dominated by Scheme (C,S,2) for the chosen parameters. Since the source has more power than the relay, the second
hop is the bottleneck. Therefore, interference management in the second
hop becomes critical and the relay should be informed about the interference. Also, for this scenario digital interference sharing performs better than compressed interference sharing. Comparing Schemes (D,U) and (D,S), we observe that while in the no fading case (D,S) is always inferior to (D,U), under fading this is no longer true and Scheme (D,S) outperforms (D,U) for large interference power.
\vspace{-0.4cm}

\section{Conclusion}

\label{conc} A relay channel with orthogonal components that is corrupted by a
single external interferer is studied. The interference is non-causally
available only at the source, but not at the relay or at the destination. The
interference is assumed to be structured, since it corresponds to a codeword
of the codebook of the interferer, whose transmission strategy is assumed to
be fixed. We complement previous work that studied the model under the
assumption of unstructured interference by establishing achievable schemes
that leverage the interference structure. Effective interference management
calls, on the one hand, for appropriate communication strategies towards the
relay in order to enable cooperative interference management, and, on the
other, for the design of joint encoding/decoding strategies. Our works sheds
light on the optimal design for DMC and AWGN\ channels with and without
fading. The best available transmission strategies turn out to depend
critically on the parameters of the interference signal (such as interference
power and transmission rate) and on the channel model.
\vspace{-0.6cm}
\appendices

\section{Proof of Proposition \ref{converse}}
\vspace{-0.4cm}
\label{appendix_converse}
From Fano's inequality, we have $H(W|Y_{D1}^{n},Y_{D2}^{n})\leq n\epsilon_{n}$,
where $\epsilon_{n}\rightarrow0$ as $n\rightarrow\infty,$ if the probability
of error goes to zero as $n\rightarrow\infty$, and thus
\vspace{-0.6cm}
\begin{subequations}
\begin{align}
nR &  \leq I(W;Y_{D_{1}}^{n}Y_{D_{2}}^{n})+n\epsilon_{n}\\
&  =I(W;Y_{D_{1}}^{n}|Y_{D_{2}}^{n})+I(W;Y_{D_{2}}^{n})+n\epsilon_{n} \label{App_C_2}\\
&  =h(Y_{D_{1}}^{n}|Y_{D_{2}}^{n})-h(Y_{D_{1}}^{n}|Y_{D_{2}}^{n}%
,W)+h(Y_{D_{2}}^{n})-h(Y_{D_{2}}^{n}|W)+n\epsilon_{n}\\
&  \leq h(Y_{D_{1}}^{n})-h(Y_{D_{1}}^{n}|X_{SD}^{n},Y_{D_{2}}^{n}%
,W)+h(Y_{D_{2}}^{n})-h(Y_{D_{2}}^{n}|X_{SR}^{n},W)+n\epsilon_{n} \label{App_C_4}\\
&  =h(Y_{D_{1}}^{n})-h(Y_{D_{1}}^{n}|X_{SD}^{n})+h(Y_{D_{2}}^{n})-h(Y_{D_{2}%
}^{n}|X_{SR}^{n})+n\epsilon_{n} \label{App_C_5} \\
&  \leq nI(X_{SD};Y_{D_{1}})+nI(X_{SR}^{n};Y_{D_{2}}^{n})+n\epsilon_{n} \label{App_C_6}%
\end{align}
\end{subequations}
\vspace{-1.2cm}\newline
where we have used the chain rule of mutual information in (\ref{App_C_2}), the fact that conditioning reduces entropy \cite{Cover} in (\ref{App_C_4}) and
the Markov chains $(Y_{D_{2}}^{n},W)-X_{SD}^{n}-Y_{D_{1}}^{n}$ and
$W-X_{SR}^{n}-Y_{D_{2}}^{n}$ in (\ref{App_C_5}). In (\ref{App_C_6}), we used the
same steps in the standard converse of a point-to-point channel which shows
that $h(Y_{D_{1}}^{n})-h(Y_{D_{1}}^{n}|X_{SD}^{n})\leq nI(X_{SD};Y_{D_{1}})$
for $X_{SD}=X_{SD,Q}$ and $Y_{D_{1}}=Y_{D1,Q}$ with $Q$ being a uniformly
distributed random variable in the set $[1,...,n]$ \cite{Cover}. This
concludes the proof.

\pagebreak \begin{figure}[t]
\center
\includegraphics[scale=0.32]{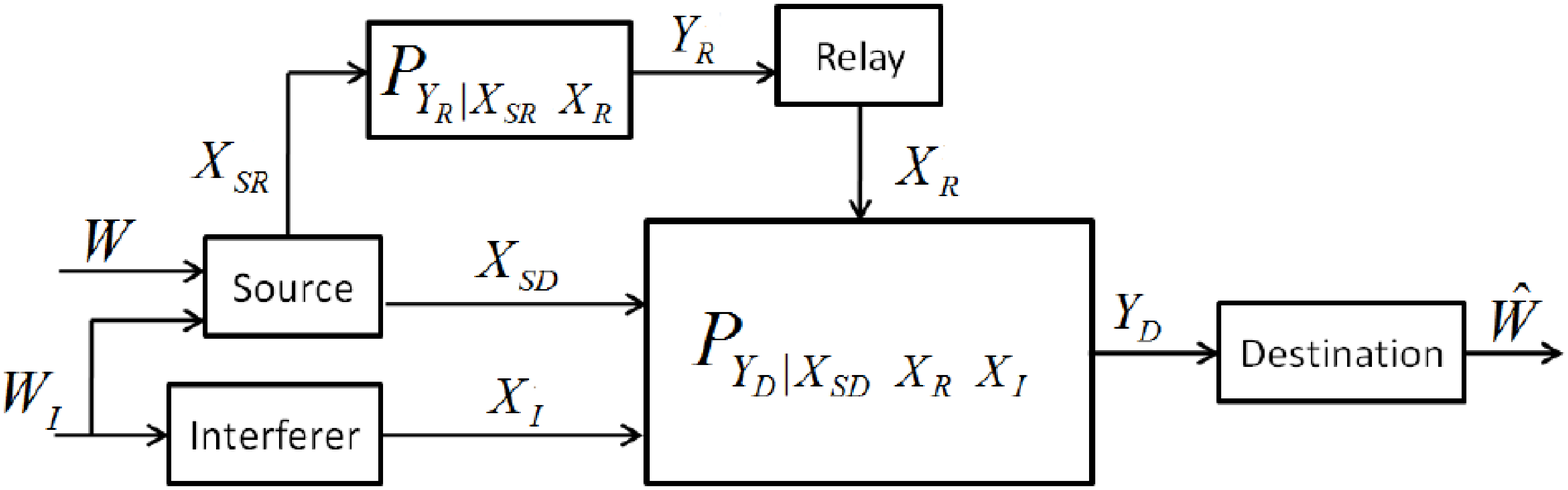} \caption{Relay channel with orthogonal components under structured
interference known at the source.}
\label{More_simplified_relay_channel_w_interference_informed_source} \end{figure}

\begin{figure}[t]
\center
\includegraphics[scale=0.26]{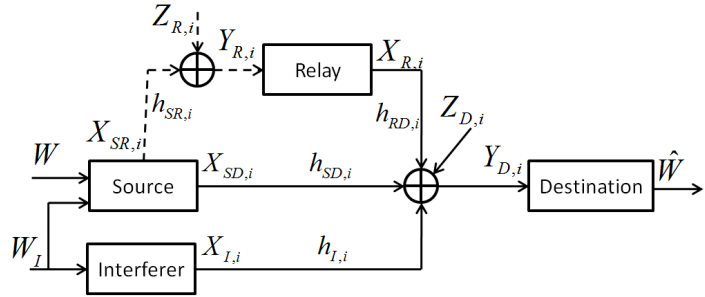} \caption{AWGN relay channel with orthogonal components under structured
interference known at the source where the dashed line denotes the out-of-band
channel between the source and the relay.} \label{More_simplified_relay_channel_w_interference_informed_source_awgn} \end{figure}

\begin{figure}[t]
\center
\includegraphics[scale=0.32]{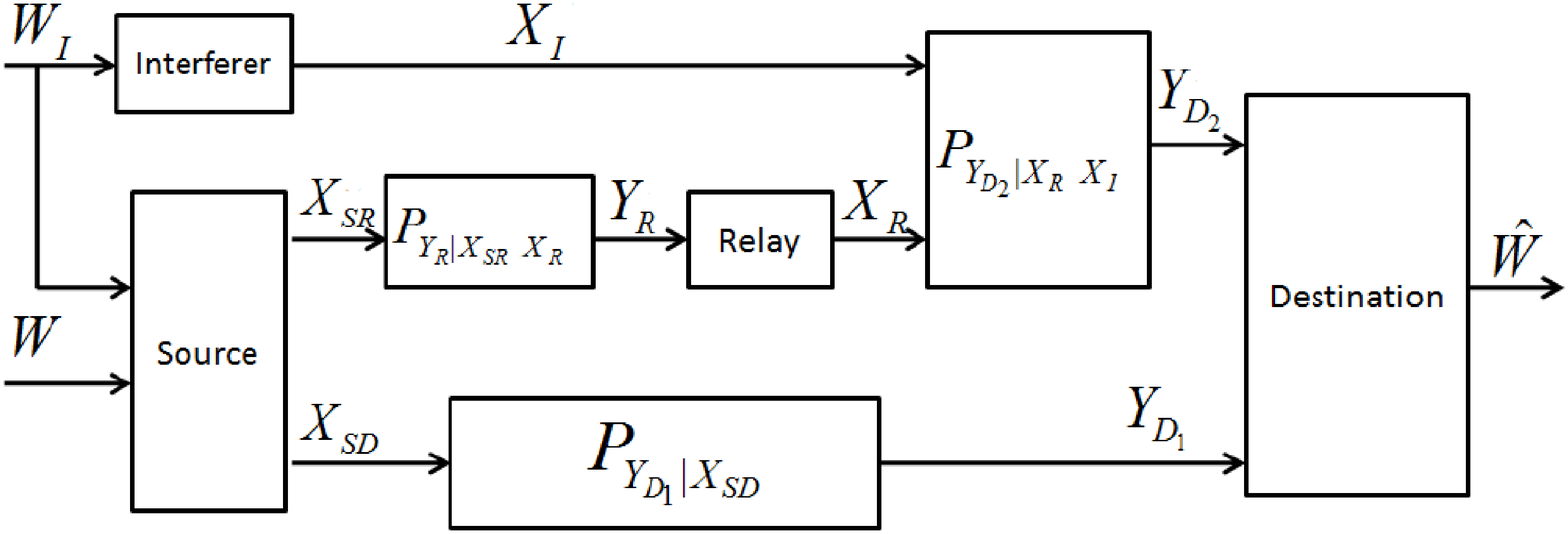} \caption{Special class of relay channel with orthogonal components under
structured interference known at the source.} \label{Special_class_relay_channel_w_interference_informed_source}\end{figure}

\begin{figure}[t]
\center
\includegraphics[scale=0.4]{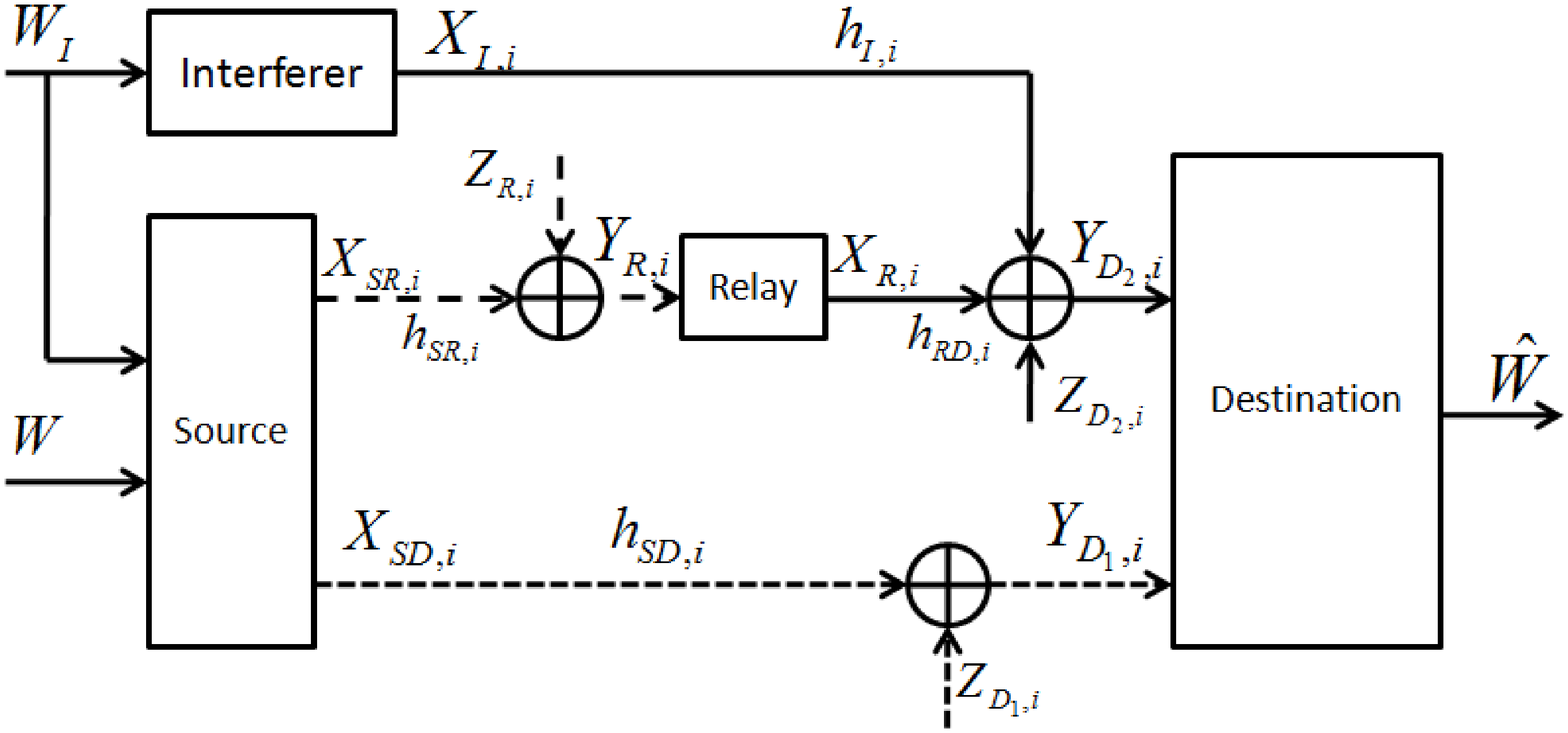} \caption{Special class of AWGN relay channel with orthogonal components under
structured interference known at the source for independent sources.}
\label{Special_class_relay_channel_w_interference_informed_source_awgn_special}\end{figure}

\begin{figure}[t]
\center
\includegraphics[scale=0.8]{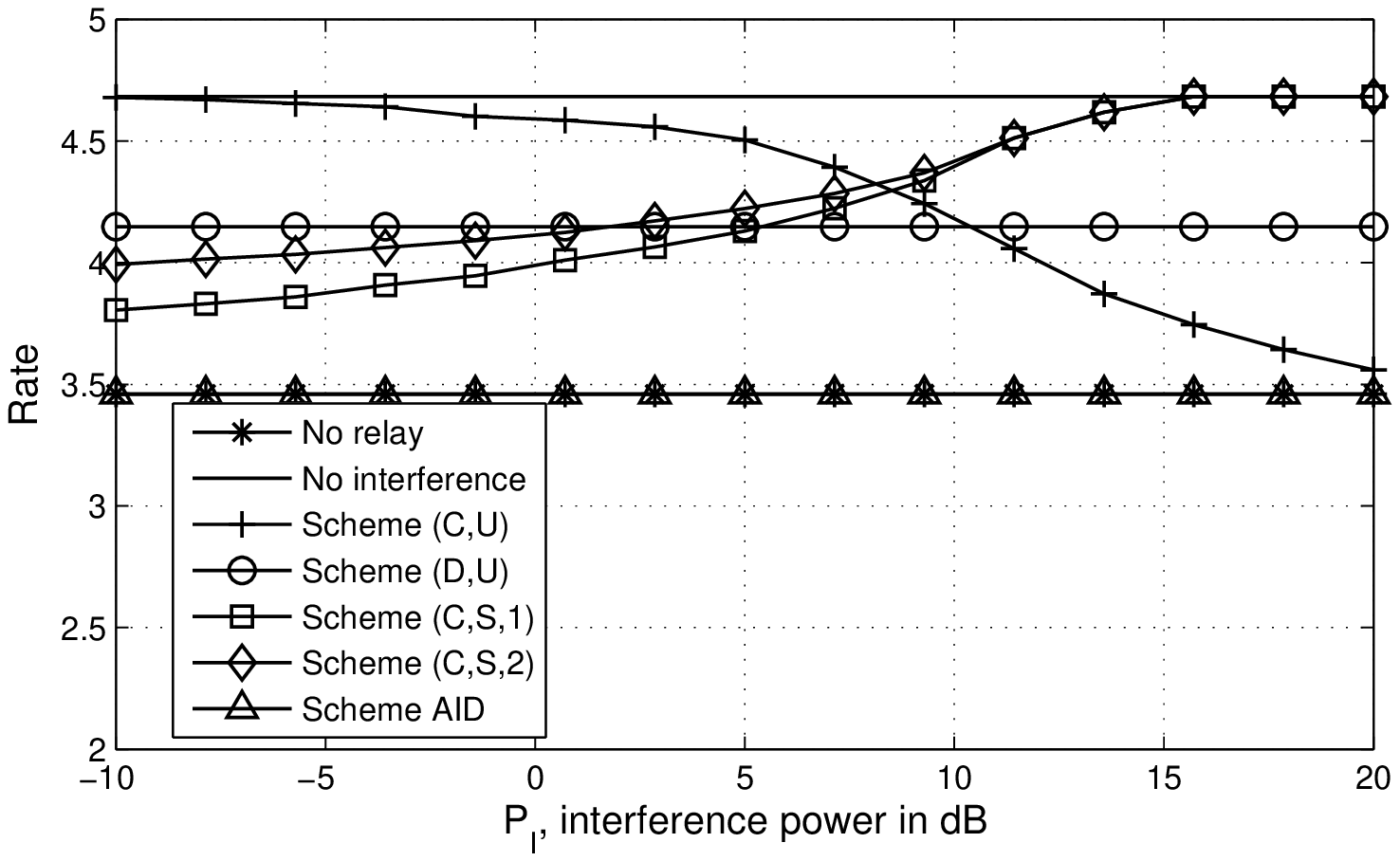} \caption{Achievable rate as a
function of $P_{I}$ when $P_{S}=P_{R}=10dB$, $|h_{SD}|=|h_{SR}|=|h_{RD}%
|=|h_{I}|=1$ and $R_{I}=1$.}%
\label{fig_1}%
\end{figure}

\begin{figure}[t]
\center
\includegraphics[scale=0.8]{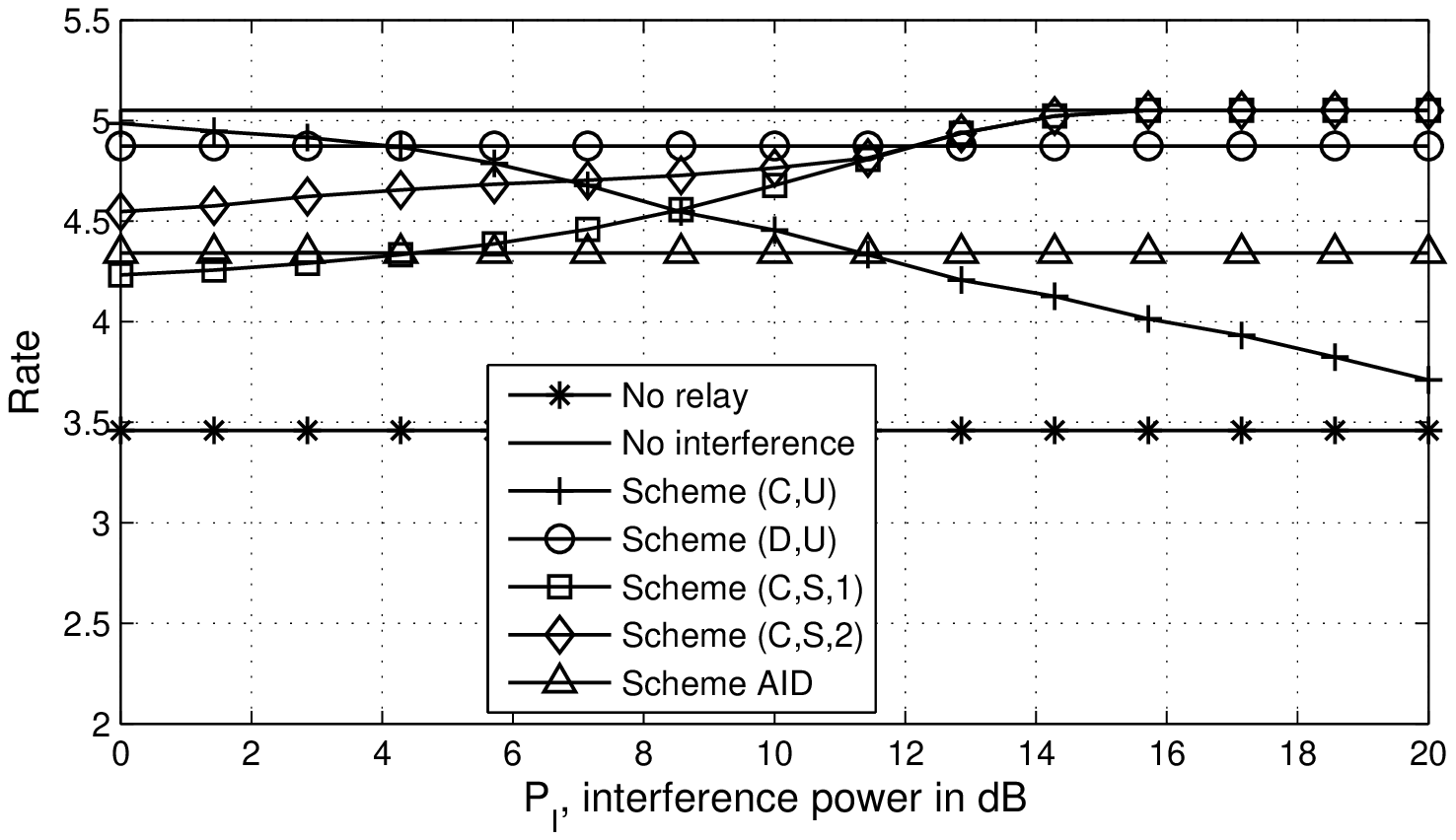} \caption{Achievable rate as a
function of $P_{I}$ when $P_{S}=P_{R}=10dB$, $|h_{SR}|=2$, $|h_{SD}%
|=|h_{RD}|=|h_{I}|=1$ and $R_{I}=1$.}%
\label{fig_2}%
\end{figure}
\begin{figure}[t]
\center
\includegraphics[scale=0.8]{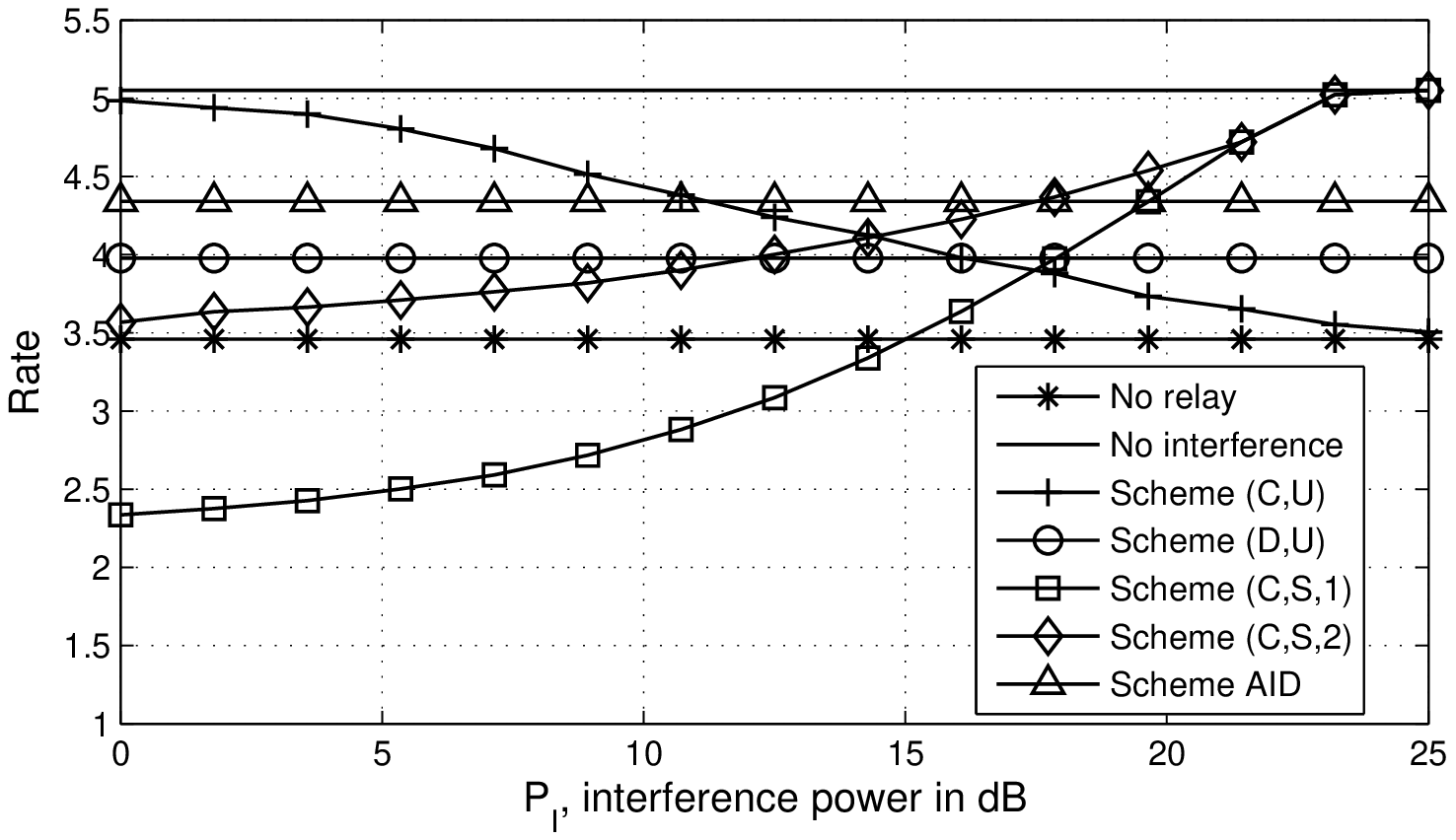} \caption{Achievable rate as a
function of $P_{I}$ when $P_{S}=P_{R}=10dB$, $|h_{SR}|=2$, $|h_{SD}%
|=|h_{RD}|=|h_{I}|=1$ and $R_{I}=3$.}%
\label{fig_3}%
\end{figure}


\begin{figure}[t]
\center
\includegraphics[scale=0.8]{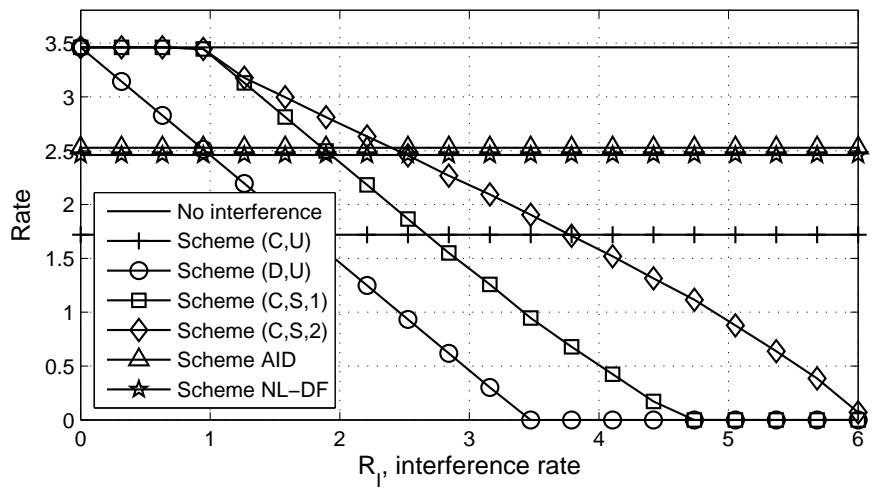} \caption{Achievable rate
as a function of $R_{I}$ for the multihop channel ($h_{SD}=0$)  when $P_{S}=P_{R}=P_{I}=10dB$, $|h_{SR}%
|=|h_{RD}|=|h_{I}|=1$.}%
\label{fig_5}%
\end{figure}

\begin{figure}[t]
\center
\includegraphics[scale=0.8]{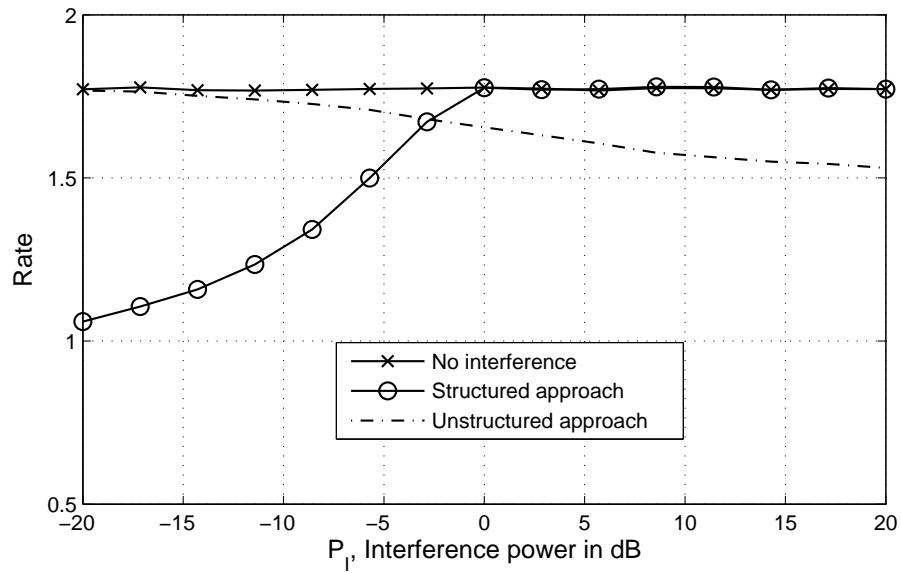} \caption{Achievable rate
as a function of $P_I$ for point to point fading channel with no
CSIT when $P_{S}=5dB$, $K=1$. }%
\label{fading_fig_2}%
\end{figure}

\begin{figure}[t]
\center
\includegraphics[scale=0.8]{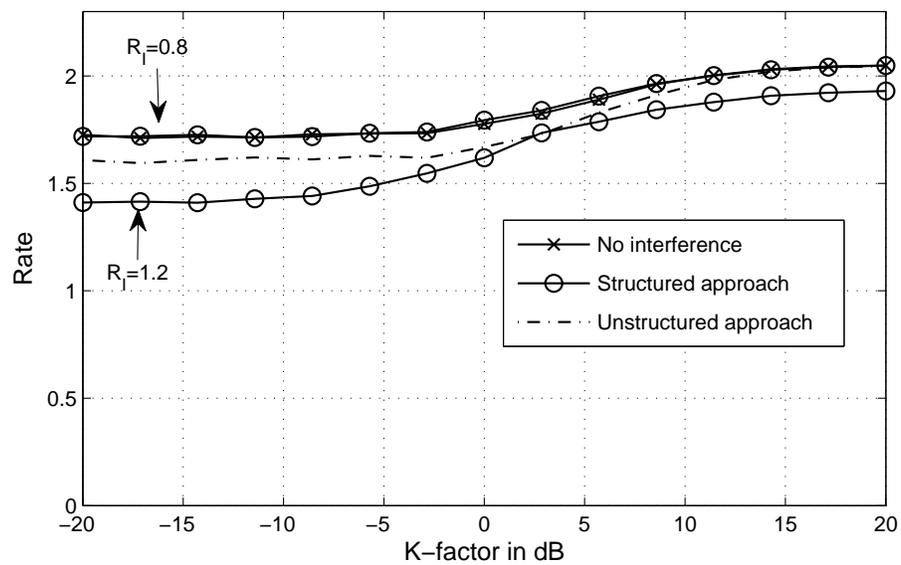}
\caption{Achievable rate as a function of K-factor for point to point fading channel with
no CSIT for various interference rates when $P_{S}=P_{I}=5dB$. }%
\label{fading_fig_1}%
\end{figure}

\begin{figure}[t]
\center
\includegraphics[scale=0.8]{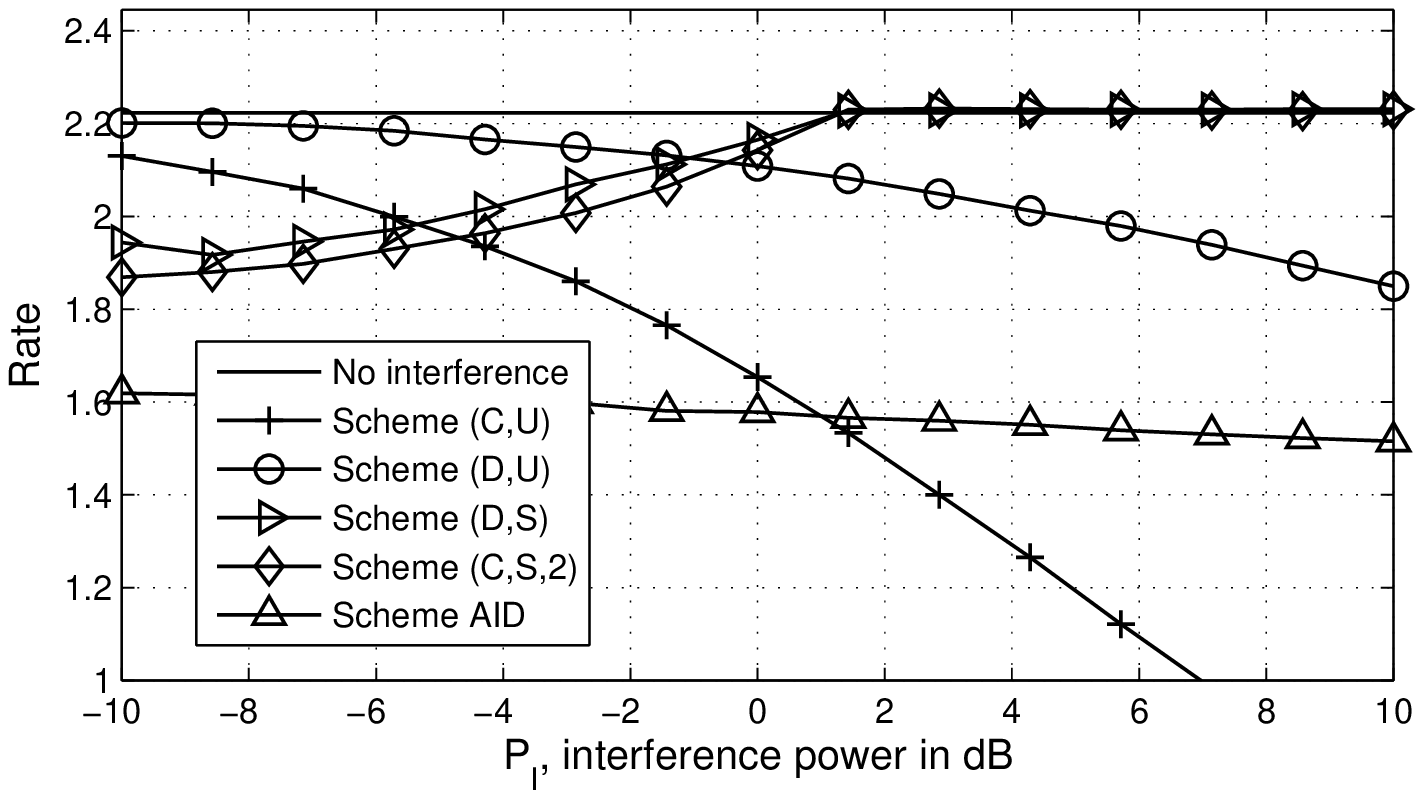}
\caption{Achievable rate as a function of $P_I$ for multihop fading channel
with no CSIT when $P_{S}=10$dB, $P_{R}=7$dB, $R_{I}=0.4$, $K=1$ and $N=1$. }%
\label{fading_fig_3}%
\end{figure}





\vspace{-0.4cm}
\begin{thebibliography}{A}



\vspace{-0.2cm}
\bibitem {Zaidi_new}A. Zaidi, S. Shamai, P. Piantanida, and L. Vandendorpe, ``Bounds on the capacity of the relay channel with noncausal state information
at source,'' arXiv:1104.1057.

\bibitem {Gelfand}S. I. Gel'fand and M. S. Pinsker, ``Coding for channel with random parameters,'' \textit{Problem of Control and Information Theory}, vol. 9, no. I, pp. 19--31, 1980.

\bibitem {heegard}C. Heegard and A. El Gamal,``On the capacity of computer memory with defects'' \textit{IEEE Trans. Inform. Theory}, vol. IT-29, no. 5, pp. 731--739, Sep. 1983.

\bibitem {Costa}M. Costa, ``Writing on dirty paper,'' \textit{IEEE Trans. Inform. Theory}, vol. IT-29, pp. 439--441, May 1983.

\bibitem {cohen}A. Cohen and A. Lapidoth , ``Generalized writing on dirty paper,'' in \textit{Proc. IEEE Int. Symp. Inform. Theory}, Lausanne, Switzerland, Jul. 2002, p. 227.

\bibitem {zamir}R. Zamir, S. Shamai and U. Erez, ``Nested linear/lattice codes for structured multiterminal binning,'' \textit{IEEE Trans. Inform. Theory}, vol. 48, no. 6, pp. 1250--1276, Jun. 2002.



\bibitem {Zhang} W. Zhang, S. Kotagiri, and J. N. Laneman, ``Writing on dirty paper with resizing and its application to quasi-static broadcast channel,'' in \textit{Proc. IEEE Int. Symp. Inform. Theory}, Nice, France, 2007, pp. 381-–385.

\bibitem {Vaze} C.S. Vaze and M. K. Varanasi, ``Dirty Paper Coding for fading channels with partial transmitter side information,'' in \textit{Proc. Asilomar Conference on Signals, Systems and Computers}, Oct. 26--29, 2008, pp. 341--345.
     \bibitem {Bennatan} A. Bennatan and D. Burshtein, ``On the Fading-Paper Achievable Region of the Fading MIMO Broadcast Channel,'' \textit{IEEE Trans. Inform. Theory}, vol. IT-54, no. 1, pp. 100--115, Jan. 2008.

\bibitem {mitran} P. Mitran, N. Devroye, and V. Tarokh, ``On compound channels with side information at the transmitter,'' \textit{IEEE Trans. Inform. Theory}, vol. IT-52, no. 4, pp. 1745–-1755, Apr. 2006.

\bibitem {Khina} A. Khina and U. Erez, ``On robust dirty paper coding,'' in \textit{Proc. IEEE Inform. Theory Workshop}, pp. 204--208, 5-9 May 2008.

\bibitem {gelfand2}S. I. Gel'fand and M. S. Pinsker, ``On Gaussian channel with random parameters,'' in \textit{Proc. IEEE Int. Symp. Inform. Theory}, Tashkent, U.S.S.R., 1984, p. 247--250.

\bibitem {9}Y.-H. Kim, A. Sutivong, and S. Sigurjonsson, ``Multiple user writing on dirty paper,'' in \textit{Proc. IEEE Int. Symp. Inform. Theory}, Chicago-USA, Jun. 2004, p. 534.

\bibitem {thesis}S. Sigurjonsson, ``Multiple user channels with state information,'' Ph.D. Thesis, Stanford University, May 2006.

\bibitem {1}A. Somekh-Baruch, S. Shamai (Shitz), and S. Verdu, ``Cooperative multiple-access encoding with states available at one transmitter,'' \textit{IEEE Trans. Inf. Theory}, vol. IT-54, no. 10, pp. 4448--4469, Oct. 2008.

\bibitem {Zaidi_EU}A. Zaidi and L. Vandendorpe, ``Lower bounds on the capacity of the relay channel with states at the source,'' \textit{EURASIP Journal on Wireless Commnunications and Networking}, vol. Article ID 634296. doi:10.1155/2009/634296, 2009.

\bibitem {Zaidi_orthogonal}A. Zaidi, S. Shamai, P. Piantanida, and L. Vandendorpe, ``On the capacity of a class of relay channels with orthogonal components and noncausal state information at source,'' \textit{ISWPC}, May 2010.

    \bibitem{Kagan_10} K. Bakanoglu and E. Erkip, ``Relay channel with structured interference known at the source,'' in \textit{Proc. Asilomar Conference on Signals, Systems and Computers}, Nov. 7--10, 2010, pp.1186--1190.

\bibitem {maric}I. Maric, N. Liu and A. Goldsmith, ``Encoding against an interferer's codebook,'' in \textit{Proc. Allerton Conference on Communications, Control and Computing}, Monticello, IL, Sept. 2008.

\bibitem {2}O. Sahin and E. Erkip, ``Cognitive relaying with one-sided interference,'' in \textit{Proc. Asilomar Conference on Signals, Systems and Computers}, Oct. 26--29, 2008,
pp. 689--694.

\bibitem {osvaldo}O. Simeone, E. Erkip, and S. Shamai,``On exploiting the interference structure for reliable communications,'' in \textit{Proc. Information Sciences and Systems}, Mar., 2010.
    \bibitem{liu} N. Liu, I. Maric, Y. Cheng, A. J. Goldsmith, and S. Shamai, ``Capacity bounds and exact results for the cognitive Z-interference channel'' arXiv:1112.2483.



\bibitem {ElGamal}A. El Gamal, and S. Zahedi, ``Capacity of a class of relay channels with orthogonal components,'' \textit{IEEE Trans. Inf. Theory}, vol. 51, no. 5, pp. 1815--1817, May 2005.


\bibitem {Cover}T. Cover and J. Thomas, \textit{Elements of Information Theory}, New York: Wiley, 1991.

\bibitem{ElGamal2} A. El Gamal and Y. H. Kim, Lecture Notes on Network Information Theory 2010 [Online]. Available: http://arxiv.org/abs/1001.3404/


\bibitem {6}T. Han, ``The capacity region of general multiple-access channel with certain correlated sources,'' \textit{Inform. Contr.}, vol. 40, no.1, pp. 37--60, 1979.

\bibitem {wolf}D. Slepian and J. K. Wolf, ``A coding theorem for multiple access channels with correlated sources'' \textit{Bell Syst. Tech. J.}, vol. 52, pp. 1037-- 1076, 1973.






\bibitem {deniz}D. Gunduz and O. Simeone, ``On the capacity region of a multiple access channel with common messages,'' in \textit{Proc. IEEE Int. Symp. Information Theory}, Austin, TX, June 2010, pp. 470--474.

\bibitem {song}Y. Song and N. Devroye, `` Structured interference-mitigation in two-hop networks,'' arXiv:1102.0964.
%
%



\bibitem {Baccelli}F. Baccelli, A. El Gamal, and D. Tse, ``Interference networks with point-to-point codes,'' \textit{IEEE Trans. Inf. Theory}, vol. 57, no. 5, pp. 2582--2596, May 2011.
\end{thebibliography}
\end{document}